%% file: main.tex
\documentclass[letterpaper]{article} 
\usepackage{aaai24}  
\usepackage{times}  
\usepackage{helvet}  
\usepackage{courier}  
\usepackage[hyphens]{url}  
\usepackage{graphicx} 
\usepackage{natbib}  
\usepackage{caption} 
\usepackage{algorithm}
\usepackage{algorithmic}
\usepackage{mathtools}
\usepackage{amsthm}
\usepackage{amssymb}

\usepackage{amsfonts}
\usepackage{bm}
\usepackage{array}
\usepackage{subcaption}
\newcolumntype{P}[1]{>{\raggedright\arraybackslash}p{#1}}

\usepackage[export]{adjustbox}

\newcommand{\name}{\textsc{MaxRefactor}}
\newcommand{\knorf}{\textsc{Knorf}}

\theoremstyle{definition}
\newtheorem{definition}{Definition}
\newtheorem{theorem}{Theorem}
\newtheorem{lemma}{Lemma}
\newtheorem{proposition}{Proposition}
\newtheorem{observation}{Observation}
\newtheorem{example}{Example}

\usepackage{newfloat}
\usepackage{listings}
\DeclareCaptionStyle{ruled}{labelfont=normalfont,labelsep=colon,strut=off} 
\floatstyle{ruled}
\newfloat{listing}{tb}{lst}{}
\floatname{listing}{Listing}
\usepackage{xcolor}

\usepackage{xcolor}
\usepackage{tikz}
\usetikzlibrary{arrows.meta}
\usepackage{pgfplots}

\title{Scalable Knowledge Refactoring using Constrained Optimisation}
\author {
    Minghao Liu\textsuperscript{\rm 1},
    David M. Cerna\textsuperscript{\rm 2},
    Filipe Gouveia\textsuperscript{\rm 1},
    Andrew Cropper\textsuperscript{\rm 1}
}
\affiliations {
    \textsuperscript{\rm 1}University of Oxford, United Kingdom\\
    \textsuperscript{\rm 2}Czech Academy of Sciences Institute of Computer Science (CAS ICS), Prague, Czechia\\
    minghao.liu@cs.ox.ac.uk,
    dcerna@cs.cas.cz,
    filipe.gouveia@cs.ox.ac.uk,
    andrew.cropper@cs.ox.ac.uk
}

\usepackage{bibentry}

\begin{document}

\maketitle

\begin{abstract}
Knowledge refactoring compresses a logic program by introducing new rules.
Current approaches struggle to scale to large programs.
To overcome this limitation, we introduce a constrained optimisation refactoring approach.
Our first key idea is to encode the problem with decision variables based on literals rather than rules.
Our second key idea is to focus on linear invented rules.
Our empirical results on multiple domains show that our approach can refactor programs quicker and with more compression than the previous state-of-the-art approach, sometimes by 60\%.

\end{abstract}

\input{01-intro}

\input{02-related}

\input{03-problem}

\input{04-method}

\input{05-exp}

\section{Conclusion and Limitations}

We have introduced a novel approach to solve the optimal knowledge refactoring problem, which refactors a logic program to reduce its size.
We implemented our ideas in \name{}, which formulates this problem as a COP.
Our experiments show that \name{} drastically reduces the sizes of refactored programs compared to the state-of-the-art approach.
Our results also show that \name{} exhibits scalability on different programs from real-world problems.

\paragraph{Limitations}
We do not support recursive invented rules or including other invented predicates in an invented rule.
Doing so may lead to better refactoring.
Future work should address this limitation.

\bibliography{main}

\newpage
\appendix
\input{06-appendix}

\end{document}

%% file: 01-intro.tex
\section{Introduction} \label{sec:intro}

Knowledge refactoring is a key component of human intelligence \cite{Rumelhart1976AccretionTA}.
Knowledge refactoring is also important in AI, such as for policy reuse in planning \cite{DBLP:journals/corr/abs-2403-16824} 
and to improve performance in program synthesis \cite{dreamcoder,dumancic2021knowledge}. 
The goal is to compress a knowledge base, such as a logic program, by introducing new rules.

To illustrate knowledge refactoring, consider the logic program:
\[
\mathcal{P}_1 = \left\{
\begin{array}{l}
     \emph{g(A)} \leftarrow \emph{p(A), q(A,B), r(B), s(A,B)} \\
     \emph{g(A)} \leftarrow \emph{p(A), q(A,B), r(B), t(A,B)} \\
     \emph{g(A)} \leftarrow \emph{p(B), q(B,C), r(C), w(A,B)} \\
    \emph{g(A)} \leftarrow \emph{p(A), q(B,A), r(A), z(A,B)}
\end{array}
\right\}
\]
This program has 4 rules, each with 5 literals. 
Thus, the size of this program is 20 (literals).

We could add a new $\emph{aux}_1$ rule to refactor $\mathcal{P}_1$ as:
\[
\mathcal{P}_2 = \left\{
\begin{array}{l}
     \emph{aux$_1$(A,B) $\leftarrow$ {p}(A), {q}(A,B), {r}(B)} \\
     \emph{g(A) $\leftarrow$ {aux}$_1$(A,B), {s}(A,B)} \\
    \emph{g(A) $\leftarrow$ {aux}$_1$(A,B), {t}(A,B)} \\
    \emph{g(A) $\leftarrow$ {aux}$_1$(B,C), {w}(A,B)} \\
    \emph{g(A) $\leftarrow$ {p}(A), {q}(B,A), {r}(A), {z}(A,B)}
\end{array}
\right\}
\]
The size of $\mathcal{P}_2$ is smaller than $\mathcal{P}_1$ (18 vs 20 literals) yet syntactically equivalent to $\mathcal{P}_1$ after unfolding\footnote{Unfolding means replacing all \emph{aux}$_1$ literals in the body with its definition literals and eliminating \emph{aux}$_1$ from the program. A formal description is in Section~\ref{sec:problem}.} \cite{tamaki1984unfold}.

A limitation of current refactoring approaches is scalability \cite{dreamcoder,babble,stitch}.
For instance, \knorf{} \cite{dumancic2021knowledge} frames the refactoring problem as a constrained optimisation problem (COP) \cite{rossi2006handbook}.
\knorf{} builds a set of invented rules  by enumerating all subsets of the rules in an input program.
\knorf{} then uses a COP solver to find a subset of the invented rules that maximally compresses the input program.
However, because it enumerates all subsets, \knorf{} struggles to scale to programs with large rules and to programs with many rules.

To overcome this scalability limitation, we introduce a novel refactoring approach.
Our first key contribution is a new COP formulation.
Instead of enumerating all subsets of rules, we use decision variables to determine whether a literal is used in an invented rule.
Our new formulation has two key advantages. 
First, the number of decision variables required is exponentially reduced.
Second, an invented rule can use any combination of literals in the input program, rather than a strict subset of an input rule.

To illustrate the benefit of our first contribution, consider the input program $\mathcal{P}_1$.
The invented rule $\emph{aux}_1$ cannot refactor the last rule {as any instance of $\emph{aux}_1$ cannot cover the literals \emph{{p}(A), {q}(B,A), {r}(A)} simultaneously}.  
However, we could invent the rule $\emph{aux}_2$ to refactor the program as:
\[
\mathcal{P}_3 = \left\{
\begin{array}{l}
    \emph{aux$_2$(A,B,C) $\leftarrow$ {p}(A), {q}(B,C), {r}(C)} \\
    \emph{g(A) $\leftarrow$ {aux}$_2$(A,A,B), {s}(A,B)} \\
    \emph{g(A) $\leftarrow$ {aux}$_2$(A,A,B), {t}(A,B)} \\
    \emph{g(A) $\leftarrow$ {aux}$_2$(B,B,C), {w}(A,B)} \\
    \emph{g(A) $\leftarrow$ {aux}$_2$(A,B,A), {z}(A,B)}
\end{array}
    \right\}
\]
The body of $\emph{aux}_2$ is not a subset of any rule in $\mathcal{P}_1$,
so 
\knorf{} could not invent $\emph{aux}_2$.
By contrast, because we allow an invented rule to use any literal, we can invent $\emph{aux}_2$.
The program $\mathcal{P}_3$ is smaller than $\mathcal{P}_2$ (16 vs 18 literals) yet is syntactically equivalent to $\mathcal{P}_1$ after unfolding.
As this example shows, our first contribution allows our approach to find better (smaller) refactorings.

Our second key contribution is to use \emph{linear invented rules} where (i) the body literals may not occur in the input program, and (ii) the size may be larger than any rule in the input program.

To illustrate this idea, consider the program:
\[
\mathcal{Q}_1 = \mathcal{P}_1 \cup \left\{
\begin{array}{l}
    \emph{g(A) $\leftarrow$ {p}(A), {p}(B), {q}(A,B), {r}(B)} \\
    \emph{g(A) $\leftarrow$ {p}(A), {q}(A,B), {q}(B,C), {r}(C)}
\end{array}
    \right\}
\]
Although we can use \emph{aux}$_1$ and \emph{aux}$_2$ to refactor $\mathcal{Q}_1$, we can find a smaller refactoring by introducing a \emph{linear invented rule} \emph{{aux}$_3$} to refactor $\mathcal{Q}_1$ as:
\[
\mathcal{Q}_2 = \left\{
\begin{array}{l}
     \emph{{aux}$_3$(A,B,C,D,E,F,G) $\leftarrow$ {p}(A), {p}(B), {q}(C,D),} \\
     \hspace{99pt} \emph{{q}(E,F), {r}(G)} \\
     \emph{g(A) $\leftarrow$ {aux}$_3$(A,A,A,B,A,B,B), {s}(A,B)} \\
    \emph{g(A) $\leftarrow$ {aux}$_3$(A,A,A,B,A,B,B), {t}(A,B)} \\
    \emph{g(A) $\leftarrow$ {aux}$_3$(B,B,B,C,B,C,C), {w}(A,B)} \\
    \emph{g(A) $\leftarrow$ {aux}$_3$(A,A,B,A,B,A,A), {z}(A,B)} \\
    \emph{g(A) $\leftarrow$ {aux}$_3$(A,B,A,B,A,B,B)} \\
    \emph{g(A) $\leftarrow$ {aux}$_3$(A,A,A,B,B,C,C)}
\end{array}
\right\}
\]
This refactoring reduces the size from 30 $(\mathcal{Q}_1)$ to 22 $(\mathcal{Q}_2)$.

In a linear invented rule, the variables occur linearly in the body literals. 
We prove (Theorem \ref{thm:linear}) that if a program can be refactored using any invented rule then it can be refactored to the same size (or smaller) using a linear invented rule.
This contribution allows our approach to find smaller refactorings.

\subsection*{Novelty and Contributions}
The three main novelties of this paper are (i) the theoretical proof of the complexity of the optimal knowledge refactoring problem, (ii) a concise and efficient encoding of the problem as a COP, and (iii) demonstrating the effectiveness of our approach on large-scale and real-world benchmarks.

Overall, our contributions are:
\begin{itemize}
    \item We introduce the \emph{optimal knowledge refactoring} problem, where the goal is to compress a logic program by inventing rules.
    We prove that the problem is $\mathcal{NP}$-hard.
    \item We introduce \name{}, which solves the optimal knowledge refactoring problem by formulating it as a COP.
    \item We evaluate our approach on multiple benchmarks. 
    Our results show that, compared to the state-of-the-art approach, \name{} can improve the compression rate by 60\%.
    Our results also show that \name{} scales well to large and real-world programs.
\end{itemize}

%% file: 02-related.tex
\section{Related Work}

\textbf{Knowledge refactoring.}
Knowledge refactoring is important in many areas of AI \cite{DBLP:journals/corr/abs-2403-16824}, notably in program synthesis \cite{dreamcoder,stitch,babble,stevie}.
For instance, \citet{dumancic2021knowledge} show that learning from refactored knowledge can substantially improve predictive accuracies of an inductive logic programming system and reduce learning times because the knowledge is structured in a more reusable way and redundant knowledge is removed.

\textbf{Model reformulation.} There is much research on reformulating constraint satisfaction problem (CSP) models automatically \cite{o2010automated,DBLP:conf/ijcai/CharlierKLD17,vo2020reformulations}.
The main categories include implied constraints generation \cite{charnley2006automatic,bessiere2007learning}, symmetry and dominance breaking \cite{DBLP:journals/mp/Liberti12,DBLP:conf/ijcai/MearsB15}, pre-computation \cite{DBLP:journals/ai/CadoliM06}, and cross-modeling language translation \cite{drake2002automatically}.
We differ because we work with logic programs and invent rules.

\textbf{Redundancy and compression.}
Knowledge refactoring is distinct from knowledge redundancy, which is useful in many areas of AI, such as in Boolean Satisfiability \cite{heule2015clause}.
For instance, \citet{plotkin:thesis} introduced a method to remove logically redundant literals and clauses from a logical theory.
Similarly, theory compression \cite{DeRaedtLuc2008CpPp} approaches select a subset of clauses such that performance is minimally affected with respect to a cost function.
We differ from redundancy elimination and theory compression because we restructure knowledge by inventing rules.

\textbf{Predicate invention.}
We refactor a logic program by introducing predicate symbols that do not appear in the input program, which is known as predicate invention \cite{kramer:ijcai20}.
Predicate invention is a major research topic in program synthesis and inductive logic programming \cite{pedro:pi,mugg:metagold,celine:bottom,luc:mdl,tom:pi,cerna2024generalisation}.
We contribute to this topic by developing an efficient and scalable method to invent predicate symbols to compress a logic program.

\textbf{Program refactoring.}
In program synthesis, many researchers \cite{dreamcoder,stitch,babble} refactor functional programs by searching for local changes (new $\lambda$-expressions) that optimise a cost function.
We differ from these approaches because we (i) consider logic programs, (ii) use a declarative solving paradigm (COP), and (iii) guarantee optimal refactoring.
\textsc{Alps} \cite{seb:alps} compresses facts in a logic program, whereas we compress rules.

\textbf{\knorf{}}.
The most similar work is \knorf{} \cite{dumancic2021knowledge}.
Given a logic program $\mathcal{P}$ as input, \knorf{} works as follows.
For each rule $r \in \mathcal{P}$, \knorf{} enumerates every subset $s$ of $r$.
For each subset $s$ and for each combination $h$ of the variables in $s$, \knorf{} creates a new rule $\emph{aux}_{s_h}$.
\knorf{} then creates a COP problem where there is a decision variable for each $\emph{aux}_{s_h}$.
It also creates decision variables to state whether a rule  $r \in \mathcal{P}$ is refactored using $\emph{aux}_{s_h}$.
\knorf{} then uses a COP solver to find a subset of the invented rules that leads to a refactoring with maximal compression.
\knorf{} has many scalability issues.
Foremost, it enumerates all subsets of all rules and thus struggles to scale to programs with large rules and many rules.
Specifically, for a program with $n$ rules and a maximum rule size $k$, \knorf{} uses $O(n2^{k})$ decision variables.
By contrast, our \name{} approach does not enumerate all subsets of rules. 
Instead, as we describe in Section \ref{sec:impl}, we define a new rule $\emph{aux}_i$ by creating decision variables to represent whether any literal in $\mathcal{P}$ is in $\emph{aux}_i$. 
Besides, we create decision variables to state whether a rule and a literal in $\mathcal{P}$ is refactored using $\emph{aux}_i$.
Since the number of invented rules is a predefined constant, \name{} only needs $O(nk)$ decision variables.
Finally, we propose the most general refactoring, which can expand the space for invented rules, enabling \name{} to find better refactorings than \knorf{}.

%% file: 03-problem.tex
\section{Problem Setting} 
\label{sec:problem}
We focus on refactoring knowledge in the form of a logic program, specifically a {definite logic program} with the least Herbrand model semantics \cite{lloyd2012foundations}.
For simplicity, we use the term {logic program} to refer to a definite logic program.
We assume familiarity with logic programming but restate some key terms.
A {logic program} is a set of \emph{rules} of the form:
$$h \leftarrow a_1, a_2, \dots, a_m$$
where $h$ is the \emph{head literal} and $a_1, \dots, a_m$ are the \emph{body literals}.
A \emph{literal} is a predicate symbol with one or more variables. 
For example, \emph{parent(A,B)} is a literal with the predicate symbol \emph{parent} and two variables \emph{A} and  $B$.
The predicate symbol of a literal $a$ is denoted as $pred(a)$ and the set of its variables is denoted as $var(a)$.
The head literal of a rule is true if and only if all the body literals are true.
The head literal of a rule $r$ is denoted as $head(r)$ and the set of body literals is denoted as $body(r)$.
The \emph{size} of a rule $r$ is defined as $size(r) = |body(r)|+1$, which is the total number of literals.
The \emph{size} of a logic program $\mathcal{P}$ is defined as $size(\mathcal{P}) = \sum_{r \in \mathcal{P}}{size(r)}$.

\subsection{Knowledge Refactoring}
Our goal is to reduce the size of a logic program whilst preserving its semantics. 
However, checking the semantic equivalence between two logic programs is undecidable \cite{shmueli1987decidability}.
Therefore, we focus on finding syntactically equivalent refactorings.
Syntactic equivalence implies semantic equivalence but the reverse is not necessarily true.

We check syntactic equivalence by \emph{unfolding} \cite{tamaki1984unfold} refactored programs.
Unfolding is a transformation in logic programming.
We adapt the unfolding definition of \citet{ilp:book} to (i) resolve multiple literals, whereas the original definition only resolves one literal each time, and (ii) prohibit variables that only appear in the body of a rule:
\begin{definition}[\textbf{Rule unfolding}]
    Given a logic program $\mathcal{P} = \{c_1, c_2, \dots, c_N\}$ and a rule $r$, where $r$ does not have variables that only appear in its body, then unfolding $\mathcal{P}$ upon $r$ means constructing the logic program $\mathcal{Q} = \{d_1, d_2, \dots, d_N\}$, where each $d_i$ is the resolvent of $c_i$ and $r$ if $head(r)$ is unifiable with any literal in $body(c_i)$, otherwise $d_i=c_i$.
\end{definition}
\begin{example}[\textbf{Rule unfolding}]
\label{example1}
Consider the program:
\[
\mathcal{P} = \left\{
\begin{array}{l}
    \emph{g(A) $\leftarrow$ {p}(A), {aux}(A,B)} \\
    \emph{g(A) $\leftarrow$ {p}(B), {p}(C), {aux}(A,B), {aux}(A,C)} \\
    \emph{g(A) $\leftarrow$ {p}(B), {q}(A,B), {r}(B)}
\end{array}
\right\}
\]
and the rule $r:\emph{aux(A,B) $\leftarrow$ {p}(B), {q}(A,B)}$. The result of unfolding $\mathcal{P}$ upon $r$ is:
\[
\mathcal{Q} = \left\{
\begin{array}{l}
    \emph{g(A) $\leftarrow$ {p}(A), {p}(B), {q}(A,B)} \\
    \emph{g(A) $\leftarrow$ {p}(B), {p}(C), {q}(A,B), {q}(A,C)} \\
    \emph{g(A) $\leftarrow$ {p}(B), {q}(A,B), {r}(B)}
\end{array}
\right\}
\]
\end{example}
\noindent
Given a set of rules $S$, $\textsf{unfold}(\mathcal{P},S)$ is the result of successively unfolding $\mathcal{P}$ upon every rule in $S$.

To refactor a logic program, we want to introduce invented rules:
\begin{definition}[\textbf{Invented rule}]
    \label{def:aux}
    Given a logic program $\mathcal{P}$ and a finite set of variables $\mathcal{X}$, an invented rule $r$ satisfies four conditions:
    \begin{enumerate}
        \item $pred(head(r)) \notin \mathcal{P}$
        \item $\forall a \in body(r),\;pred(a) \in \mathcal{P}$
        \item $\forall a \in body(r),\;var(a) \subseteq \mathcal{X}$
        \item $var(head(r)) = \bigcup_{a \in body(r)}{var(a)}$
    \end{enumerate} 
\end{definition}

\noindent 
In this paper, we use the prefix \emph{aux} to denote the predicate symbol of an invented rule.

We want to refactor an input rule using invented rules.
We call such rules \emph{refactored rules}:
\begin{definition}[\textbf{Refactored rule}]
    Let $r$ be a rule and $S$ be a set of invented rules.
    Then $r'$ is a refactored rule of $r$ iff $\textsf{unfold}(\{r'\},S)=\{r\}$.
    For a logic program $\mathcal{P}$ and a set of invented rules $S$, $\mathcal{I}(\mathcal{P},S)$ denotes the set of all possible refactored rules of rules in $\mathcal{P}$.
\end{definition}

\begin{example}[\textbf{Refactored rule}]
\label{example2}
    Given the rules in $\mathcal{Q}$ and the invented rule \emph{aux} in Example~\ref{example1}, the first two rules in $\mathcal{P}$ are refactored rules.
\end{example}

\noindent
We define a \emph{refactored program}:
\begin{definition}[\textbf{Refactored program}]
    \label{def:rp}
    Given a logic program $\mathcal{P}$ and a finite set of invented rules $S$, a logic program $\mathcal{Q} \subseteq \mathcal{P} \cup S \cup \mathcal{I}(\mathcal{P},S)$ is a refactored program of $\mathcal{P}$ iff $\textsf{unfold}(\mathcal{Q} \setminus S, \mathcal{Q} \cap S) = \mathcal{P}$, and we refer to $\mathcal{Q}$ as a \emph{proper refactoring} if $\forall S' \subset \mathcal{Q} \cap S$, $\textsf{unfold}(\mathcal{Q} \setminus S, \mathcal{Q} \cap S') \neq \mathcal{P}$.
\end{definition}

\noindent 
To help explain the intuition behind Definition~\ref{def:rp}, note the following three properties of a refactored program.
First, a refactored program consists of three kinds of rules (a) input rules from $\mathcal{P}$, 
(b) invented rules from $S$, 
and (c) refactored rules from $\mathcal{I}(\mathcal{P},S)$.
Second, $\mathcal{P}$ and $\mathcal{Q}$ must be syntactically equivalent, which means that we can unfold $\mathcal{Q}$ upon the invented rules to get $\mathcal{P}$.
Third, we are only interested in programs $\mathcal{Q}$ which are proper refactorings of $\mathcal{P}$ as any non-proper refactoring containing $\mathcal{Q}$ is always a larger program.

We want \emph{optimal knowledge refactoring (OKR)}:
\begin{definition}[\textbf{Optimal knowledge refactoring problem}]
    \label{def:oprp}
    Given a logic program $\mathcal{P}$ and a finite set of invented rules $S$, the {optimal knowledge refactoring} problem is to find a refactored program $\mathcal{Q}$ such that for any refactored program $\mathcal{Q'}$, $size(\mathcal{Q}) \le size(\mathcal{Q'})$.
\end{definition}

\noindent We prove that the OKR problem is $\mathcal{NP}$-hard:
\begin{theorem}
\label{thm:complex}
    The optimal knowledge refactoring problem is $\mathcal{NP}$-hard.
\end{theorem}
\begin{proof}[Proof (Sketch)] 
{The full proof is in \appendixname~\ref{app:proofcomplex}. We provide a reduction from \emph{maximum independent set in 3-regular Hamiltonian graphs}~\cite{FLEISCHNER20102742} to a propositional variant of OKR where $c_1,c_2\in {S}$ may refactor the same rule of $\mathcal{P}$ iff $\mathit{body}(c_1)\cap \mathit{body}(c_2)= \emptyset$ (OPKR$_{WD}$). Let $G= (V,E)$ be a 3-regular Hamiltonian graph and $C$ a maximum independent set of $G$. Observe that in $G'= (V\setminus C,E')$, where $E'$ contains all edges from $E$ with endpoints in $V\setminus C$, every vertex has degree at most 2. Thus, computing a maximum weighted independent set of $G'$ is possible in polynomial time. Furthermore, it can be shown that optimal refactorings of instances derived from a 3-regular Hamiltonian graph $G$ are always a superset of a maximum independent set of $G$. Thus, we can reduce OPKR$_{WD}$ to OPKR, a fragment of OKR.}
\end{proof}

%% file: 04-method.tex
\section{\name{}} \label{sec:impl}

Given an input logic program $\mathcal{P}$, a space of possible invented rules $S$, and a maximum number $K$ of invented rules to consider, \name{} refactors $\mathcal{P}$ by finding ${C} \subseteq {S}$ where $|{C}|=K$.
In other words, \name{} solves the optimal knowledge refactoring problem.

Before describing how we search for a refactoring, we first introduce \emph{linear invented rules} to restrict the space of possible invented rules.

\subsection{Linear Invented Rule}
The space of possible invented rules ${S}$ directly influences the effectiveness and efficiency of refactoring.
For example, \knorf{} defines ${S}$ as all subsets of rules in an input program.
For instance, if there is an input rule \emph{{g}(A) $\leftarrow$ {p}(A,B), {q}(B,C), {r}(B,D)}, \knorf{} adds four possible invented rules to ${S}$:
\[
\left\{
\begin{array}{l}
    \emph{aux$_1$(A,B,C) $\leftarrow$ {p}(A,B), {q}(B,C)} \\
    \emph{aux$_2$(A,B,D) $\leftarrow$ {p}(A,B), {r}(B,D)} \\
    \emph{aux$_3$(B,C,D) $\leftarrow$ {q}(B,C), {r}(B,D)} \\
    \emph{aux$_4$(A,B,C,D) $\leftarrow$ {p}(A,B), {q}(B,C), {r}(B,D)}
\end{array}
\right\}
\]
Defining ${S}$ this way has the advantage that it restricts which literals can appear in an invented rule.
However, this approach also has disadvantages. 
As we show in Section~\ref{sec:intro}, this approach cannot invent a rule that is not a subset of any input rule (e.g. \emph{aux}$_2$ in $\mathcal{P}_3$), nor can the size of the invented rule be larger than any input rule (e.g., \emph{aux}$_3$ in $\mathcal{Q}_2$).

To overcome these limitations, we use a new approach to define the space of invented rules.
We denote the set of body predicate symbols in the logic program $\mathcal{P}$ as $Pr({\mathcal{P}})$ and a finite set of variables as $\mathcal{X}$.
We define the set of all possible body literals $L = \left\{ p(v) \mid p \in Pr({\mathcal{P}}), v \in \mathcal{X}^{\mathit{arity}(p)} \right\}$.
We define ${S}$ as the space of all combinations of literals from $L$, i.e., the power set of $L$.
The size of ${S}$ is exponential in the size of $L$ and $S$ is clearly a superset of the space considered by \knorf{}.
Therefore, our key insight is not to consider all combinations of $L$ and to instead only consider the most general ones.
We call such rules \emph{linear invented rules}:
\begin{definition}[\textbf{Linear invented rule}]
    A linear invented rule is rule with linear variable occurrence in its body.
\end{definition}
\begin{example}[\textbf{Linear invented rule}]
The rule \emph{aux$_1$(A,B,C) $\leftarrow$ {p}(A,B), {q}(B,C)} is not linear because its body literals share the variable $B$. 
By contrast, the rule \emph{aux$_2$(A,B,C,D) $\leftarrow$ {p}(A,B), {q}(C,D)} is linear.
\end{example}

\noindent
For any invented rule $r$, we can always build a linear invented rule $lin(r)$ as follows:
(i) rewrite the body literals such that the predicate symbols remain unchanged and the variables occur linearly, then (ii) build a new head literal such that $var(head(lin(r))) = \bigcup_{a \in body(lin(r))}{var(a)}$.

To motivate the use of linear invented rules in refactoring, we show the theorem:
\begin{theorem}
    \label{thm:linear}
    If the optimal knowledge refactoring problem has a solution using the set of invented rules ${C} \subseteq {S}$ then it has a solution using ${C}' \subseteq {S}_{lin}$, where ${S}_{lin} = \left\{lin(s) \mid s \in S \right\}$.
\end{theorem}
\begin{proof}
  Suppose there is an invented rule $r \in {C}$, we can always build a linear invented rule $r'=lin(r)$, then ${C}'={C} \setminus \{r\} \cup \{r'\}$.
  First, observe that $|{C}'|\leq|{C}|$ because either an invented rule has a unique linear invented rule generalising it, or multiple invented rules of $C$ are generalised by the same linear invented rule.
  Second, for any body literal in the input program that is covered by $r$, it can also be covered by $r'$ via variable mapping.
  Thus, we can always get a refactored program of the same or smaller size using ${C}'$.
\end{proof}

\noindent
{Theorem~\ref{thm:linear} implies that we can restrict $S$ to linear invented rules without compromising the optimal solution.
As we show below, defining $S$ as linear invented rules allows us to simplify the COP formulation.}

\subsection{COP Encoding}
\label{sec:encoding}

According to Theorem~\ref{thm:complex}, the optimal knowledge refactoring problem is $\mathcal{NP}$-hard, making it computationally challenging.
Constraint programming (CP) is a successful framework for modeling and solving hard combinatorial problems \cite{DBLP:conf/ijcai/Hebrard18}.
Therefore, \name{} formulates this search problem as a constrained optimisation problem (COP) \cite{rossi2006handbook}.
Given (i) a set of \emph{decision variables}, (ii) a set of \emph{constraints}, and (iii) an \emph{objective function}, a COP solver finds an assignment to the decision variables that satisfies all the specified constraints and minimises the objective function.

We describe our COP encoding below\footnote{
COP problems can also be encoded as MaxSAT problems \cite{maxsat}, an optimisation variant of SAT.
We also develop a MaxSAT encoding and include it in the appendix.}.

\subsubsection{Decision Variables}
\name{} builds decision variables to determine (a) which literals are in which invented rules, and (b) how to refactor input rules using the invented rules\footnote{
\knorf{} is not involved in task (a), as it creates a decision variable for each invented rule in $S$.
However, as previously stated, in our definition the size of $S$ is larger than \knorf{}, even when considering only linear invented rules.
To improve scalability, \name{} uses decision variables based on literals to build invented rules.
Specifically, we only find which predicate symbols occur in the body of an invented rule, then assign variables linearly.}.

For task (a), for each possible invented rule $r_k$ for $k \in [1,K]$ and $p \in Pr({\mathcal{P}})$, we use an integer variable $\mathit{r}_{k,p}$ to indicate the number of literals with predicate symbol $p$ in the body of $r_k$. 

For example, consider the input program:
\[
\mathcal{P} = \left\{
\begin{array}{l}
 c_1: \emph{g(A)} \leftarrow \emph{p(B)\textsuperscript{\rm a1}, p(C)\textsuperscript{\rm a2}, q(A,B)\textsuperscript{\rm a3}, q(B,C)\textsuperscript{\rm a4}} \\
 c_2: \emph{g(A)} \leftarrow \emph{p(B)\textsuperscript{\rm a1}, q(A,B)\textsuperscript{\rm a2}, q(A,C)\textsuperscript{\rm a3}, s(C)\textsuperscript{\rm a4}} 
\end{array}
\right\}
\]
where the superscript denotes the index of a body literal.
For simplicity, assume that $K = 1$, i.e.  we can only invent one rule denoted $r_1$. 
The decision variables for inventing this rule are $\{\mathit{r}_{1,p}$, $\mathit{r}_{1,q}$, $\mathit{r}_{1,s}\}$.
Suppose $\mathit{r}_{1,p} = \mathit{r}_{1,q} = 1$ and $\mathit{r}_{1,s} = 0$.
Then $r_1: \emph{aux}_1\emph{(A,B,C)} \leftarrow \emph{p(A), q(B,C)}$.

For task (b), \name{} creates two groups of decision variables: (i) $\mathit{use_{c,k}}$ to denote whether an input rule $c$ is refactored with the invented rule $r_k$, and (ii) $\mathit{cover_{c,a,k}}$ to denote whether a body literal $a$ in an input rule $c$ is covered by an invented rule $r_k$, which means this literal is not in the refactoring of $c$.

Concerning (i), for each input rule $c \in \mathcal{P}$ and $k \in [1,K]$ we use a Boolean variable $\mathit{use}_{c,k}$ to indicate that $c$ is refactored using the invented rule $r_k$.
However, as we show later in the example, an input rule can be refactored using the same invented rule multiple times.
Therefore, we use the Boolean variable $\mathit{use}_{c,k}^t$ where $t\geq 1$ to indicate that $r_k$ is used $t$ times\footnote{
The domain of $t$ is defined by the maximum number of times a predicate symbol appears in an input rule.
}
in the refactoring of $c$.

Concerning (ii), for each input rule $c \in \mathcal{P}$, literal $a \in body(c)$, $k \in [1,K]$, and integer $t \geq 1$, we use a Boolean variable $\mathit{cover}_{c,a,k}^{t}$ to indicate that $a$ is covered by the $t$-th instance of $r_k$ in the refactoring.

Regarding the above example, the decision variables used to determine the refactoring of each input rule are:
\[
\noindent
\begin{array}{llll}
    {\mathit{use}_{c_1,1}^1} & {\mathit{use}_{c_1,1}^2} & {\mathit{use}_{c_2,1}^1} & {\color{lightgray}\mathit{use}_{c_2,1}^2} \\
\end{array}
\]
\[
\noindent
\begin{array}{llll}
    {\mathit{cover}_{c_1,{\rm a1},1}^{1}} & {\color{lightgray}\mathit{cover}_{c_1,{\rm a1},1}^{2}} & {\color{lightgray}\mathit{cover}_{c_1,{\rm a2},1}^{1}} & {\mathit{cover}_{c_1,{\rm a2},1}^{2}} \\
    {\mathit{cover}_{c_1,{\rm a3},1}^{1}} & {\color{lightgray}\mathit{cover}_{c_1,{\rm a3},1}^{2}} & {\color{lightgray}\mathit{cover}_{c_1,{\rm a4},1}^{1}} & {\mathit{cover}_{c_1,{\rm a4},1}^{2}} \\ 
    {\mathit{cover}_{c_2,{\rm a1},1}^{1}} & {\color{lightgray}\mathit{cover}_{c_2,{\rm a1},1}^{2}} & {\mathit{cover}_{c_2,{\rm a2},1}^{1}} & {\color{lightgray}\mathit{cover}_{c_2,{\rm a2},1}^{2}} \\ 
    {\color{lightgray}\mathit{cover}_{c_2,{\rm a3},1}^{1}} & {\color{lightgray}\mathit{cover}_{c_2,{\rm a3},1}^{2}} & {\color{lightgray}\mathit{cover}_{c_2,{\rm a4},1}^{1}} & {\color{lightgray}\mathit{cover}_{c_2,{\rm a4},1}^{2}}
\end{array}
\]
If the {\color{lightgray} gray} decision variables are false and the rest are true then the refactored rules are:
\[
\mathcal{Q} = \left\{
\begin{array}{l}
\emph{g(A)} \leftarrow \emph{aux}_1\emph{(B,A,B), }\emph{aux}_1\emph{(C,B,C)} \\
\emph{g(A)} \leftarrow \emph{aux}_1\emph{(B,A,B), }\emph{q(A,C), s(C)}
\end{array}
\right\}
\]

\subsubsection{Constraints}
\name{} ensures the validity of the refactoring through a set of constraints.

We use Constraint~(\ref{eq:cons2}) to ensure that for $k \in [1,K]$,  if the invented rule $r_k$ is used $t$ times in the refactoring of $c \in \mathcal{P}$ then $r_k$ is used $t-1$ times in the refactoring of $c$.
\begin{equation}
    \label{eq:cons2}
    \forall c,k,t\,(t>1),\; \mathit{use}_{c,k}^t \rightarrow \mathit{use}_{c,k}^{t-1}
\end{equation}

\noindent
We use Constraint~(\ref{eq:cons5}) to ensure that for $k \in [1,K]$, if $r_k$ is used $t$ times to refactor $c \in \mathcal{P}$, then $r_k$ cannot contain a predicate symbol $p$ which is not the predicate symbol of at least one literal in $\mathit{body}(c)$.
\begin{equation}
\label{eq:cons5}
    \forall c, p, k, t\,(p {\rm \;not\;in\;} c),\; \neg \left( \mathit{use}_{c,k}^t \wedge \mathit{r}_{k,p}>0 \right)
\end{equation}

\noindent
We use Constraint~(\ref{eq:cons6}) to ensure that for $k\in [1,K]$, that any predicate symbol $p$ of a literal $a$ in the body of $c\in \mathcal{P}$ covered by invented rule $r_k$, is the predicate symbol of a literal in the body of $r_k$. Below, $e(a,p)$ abbreviates $pred(a)=p$.
\begin{equation}
    \label{eq:cons6}
    \forall c,a,k,t\,(e(a,p))\; \mathit{cover}_{c,a,k}^{t} \rightarrow \left( \mathit{use}_{c,k}^t \wedge \mathit{r}_{k,p}>0 \right)
\end{equation}
We use Constraint~(\ref{eq:cons7}) to ensure that for $k\in [1,K]$ and $c\in \mathcal{P}$, if the predicate symbol $p$ occurs in multiple body literals of $c$ and $r_{k,p}$ times in $r_k$ then a single use of $r_k$ can only cover $r_{k,p}$ of the occurrences of $p$ in $c$.

\begin{equation}
    \label{eq:cons7}
    \forall c,p,k,t,\; \left(\sum_{a \in body(c) \wedge e(a,p)} \mathit{cover}_{c,a,k}^{t}\right) \le \mathit{r}_{k,p}
\end{equation}
Note that if $\mathit{r}_{k,p}$ is larger than the number of body literals with predicate symbol $p$ in a rule, it would still be valid, because these literals can be refactored by duplication.
For example, given the input rule \emph{{g}(A) $\rightarrow$ {p}(A,B)} and the invented rule \emph{{aux}(A,B,C,D) $\rightarrow$ {p}(A,B),{p}(C,D)}, we can refactor the same literal twice and get \emph{{g}(A) $\rightarrow$ {aux}(A,B,A,B)}.

\subsubsection{Objective}

Before introducing our objective function, we first create two useful yet redundant decision variables.
First, for each $k \in [1,K]$, we use a Boolean variable $\mathit{used}_{k}$ to indicate whether the invented rule $r_k$ is used in the refactoring.
Second, for each input rule $c \in \mathcal{P}$ and literal $a \in body(c)$, we use a Boolean variable $\mathit{covered}_{c,a}$ to indicate that $a$ is not in the refactoring of $c$.
These variables are directly calculated from $\mathit{use}$ and $\mathit{cover}$ respectively.
The constraints are as follows:
\begin{equation}
    \label{eq:cons3}
    \forall k,\; \left( \mathit{used}_k \leftrightarrow \exists c,t,\; \mathit{use}_{c,k}^t \right)
\end{equation}
\begin{equation}
    \label{eq:cons4}
    \forall c,a,\; \left( \mathit{covered}_{c,a} \leftrightarrow \exists k,t,\; \mathit{cover}_{c,a,k}^{t} \right)
\end{equation}

\noindent
Following Definition~\ref{def:oprp}, \name{} searches for the refactored program with the smallest size. In our encoding, \name{} achieves this goal by minimising the following objective function using a COP solver:
\begin{equation}
    \label{eq:obj}
    \sum_{k}{\mathit{used}_{k}} + \sum_{k,p}{\mathit{r}_{k,p}} + \sum_{c,k,t}{\mathit{use}_{c,k}^t} - \sum_{c,a}{\mathit{covered}_{c,a}}
\end{equation}
The first two terms are the size of the head and body literals used to define the invented rules added to the refactored program respectively.
The third term is the size of the invented rules in the refactored rules of the input program.
The last term is the total number of covered body literals in the input program, which should be deducted from the size.
Note that by adding the objective to the size of the input program, we get the size of the refactored program.

\subsection{The Optimal Refactoring}
We show that by solving the aforementioned COP problem, \name{} solves the OKR problem (Definition~\ref{def:oprp}):

\begin{theorem}
    \label{thm:optref}
    Let $\mathcal{P}$ be a logic program such that $|\mathcal{P}|=n$ and ${S}$ a set of invented rules. 
    Then  \name{} can solve the OKR problem with 
        $K= \left\lceil \frac{n}{4}\right\rceil\left\lceil \frac{s-1}{2}\right\rceil$,
where $s = \max \{ \mathit{size}(c) \mid c \in \mathcal{P} \}$.
\end{theorem}

\noindent
The proof is shown in \appendixname~\ref{app:tightbound}.

%% file: 05-exp.tex
\section{Experiments}
To test our claim that \name{} can find better refactorings than current approaches, our experiments aim to answer the question:

\vspace{3pt}
\noindent
\textbf{Q1:} How does \name{} compare against \knorf{} on the optimal knowledge refactoring problem?
\vspace{3pt}

\noindent 
To answer \textbf{Q1}, we compare the compression rate of the refactored programs produced by \name{} and \knorf{} given the same timeout.

To understand the scalability of \name{}, our experiments aim to answer the question:

\vspace{3pt}
\noindent
\textbf{Q2:} How does \name{} scale on larger and more diverse input programs?
\vspace{3pt}

\noindent To answer \textbf{Q2}, we evaluate the runtime and refactoring rate of \name{} on progressively larger programs.

\paragraph{Settings}
We use two versions of \name{}: one which uses the COP solver CP-SAT\footnote{https://developers.google.com/optimization} and another which uses the MaxSAT solver UWrMaxSat
\cite{piotrow2020uwrmaxsat}.
We consider at most two invented rules in a refactoring, as we found it leads to good results.
For \knorf{}, we run the code from its repository\footnote{https://github.com/sebdumancic/knorf\_aaai21} and use the same parameters as stated in that paper.
All experiments use a computer with 5.2GHz CPUs and 16GB RAM. 
We run \name{} and \knorf{} on a single CPU.

\paragraph{Reproducibility} The experimental data and the code to reproduce the results is in the appendix and will be made publicly available if the paper is accepted for publication.

\definecolor{mygreen}{rgb}{0.1,0.5,0.1}
\definecolor{myblue}{rgb}{0.0,0.4,0.8}
\definecolor{mygray}{rgb}{0.6,0.6,0.6}
\input{data/CompressionCompareCPSATv2.tex}
\input{data/CompressionCompareMaxSATv2.tex}

\subsection{Q1: Comparison with \knorf{}}

\paragraph{Datasets} We use two datasets, Lego and Strings, from the \knorf{} paper. 
Lego has 200 programs, each with a size of 264--411. 
Strings has 196 programs, each with a size of 779--12,309.
We tried using \knorf{} on other datasets but it frequently crashed on them\footnote{\knorf{} crashes on other datasets during its COP modeling phase. 
We contacted the authors of \knorf{} for advice and they confirmed that it is a fundamental bug related to incorrect assumptions on the input program. 
We tried to fix the bug according to their advice but it fixed only a small number of crashes and the runtime significantly increased. 
Therefore, we only report the results on Lego and Strings from the \knorf{} paper.
} so we only report the results on these two datasets.

\paragraph{Method}
We compare the compression rate of the refactored programs returned by \name{} and \knorf{}.
Given an input program $\mathcal{P}$ and a refactored program $\mathcal{Q}$, the compression rate $cr$ is defined as: 
$$cr=\frac{size(\mathcal{P})-size(\mathcal{Q})}{size(\mathcal{P})}$$

\noindent
To be clear, a larger $cr$ indicates better refactoring.

\input{data/Scalability}

\paragraph{Results}
\figurename~\ref{fig:comp-cpsatv2} compares the compression rates of \name{} and \knorf{} with the same timeout. Both approaches call CP-SAT.
For Lego, \name{} can increase compression rates by at least 27\% across all programs.
For Strings, \name{} finds better refactorings on 172 out of 196 programs, increasing the compression rates by an average of 30\%.
In the remaining 24 programs, both \name{} and \knorf{} fail to find any refactoring in half of them, while in the other half \name{}'s compression rates are on average 10\% lower than \knorf{}'s.

We also use a version of \name{} that uses a MaxSAT solver (See \appendixname~\ref{app:maxsat} for encoding details).
\figurename~\ref{fig:comp-maxsatv2} shows the results when compared to \knorf{}. 
Given the same timeout (600s), \name{} with MaxSAT achieves better refactoring effect compared to CP-SAT.
For Strings, \name{} finds better refactorings than \knorf{} on all the 196 programs, increasing the compression rates by 30\%--60\% with an average of 53\%.

In conclusion, these results suggest that \name{} performs substantially better than \knorf{} at solving the knowledge refactoring problem, thereby answering \textbf{Q1}.

\subsection{Q2: Scalability}

\paragraph{Datasets.} We evaluate the scalability of \name{} on three real-world problems: DrugDrug \cite{dhami2018drug}, Alzheimer \cite{DBLP:journals/ngc/KingSS95}, and WordNet \cite{DBLP:journals/ml/BordesGWB14}.
Each dataset contains millions of rules so we can randomly sample programs of distinct sizes from them and use the sampled programs to test the performance of \name{}.
Specifically, for every size in \{200, 250, 300, \dots, 1000\}, we uniformly sample 10 programs from each dataset and run \name{} (calling MaxSAT) to refactor them.

\paragraph{Results.} \figurename~\ref{fig:time-size} shows the runtime required for \name{} to find an optimal solution and prove optimality for different sizes of input programs. The timeout for a single run is set to 3600s and the runtime is counted as 3600s if the limit is exceeded.
For a program size of 400, the average runtime on all the datasets is less than 500s.
The performance varies across different datasets. 
For DrugDrug and Alzheimer, \name{} can scale to programs with a size of 600 and find the optimal refactorings. 
For WordNet, it can scale to programs with a size of 1000.

{\figurename~\ref{fig:time-size} shows the runtime required for \name{} to solve the optimal solution for different sizes of input programs, i.e. to find an optimal refactoring and, crucially, prove optimality.}
However, in most cases, \name{} can find near-optimal refactorings quickly but takes a while to prove optimality. 
To illustrate this behaviour, we select 3 challenging programs from different datasets and measure the sizes of the best refactored programs found by \name{} under different time limits.
Compared to the runtime to prove the optimal refactorings (opt-time), \name{} can find refactored programs within 20\% of the final size with only 44\% (\figurename~\ref{fig:gap-time-ex1}),  0.2\% (\figurename~\ref{fig:gap-time-ex2}), and 0.3\% (\figurename~\ref{fig:gap-time-ex3}) of the runtime respectively.
For example, in the DrugDrug task (\figurename~\ref{fig:gap-time-ex2}), \name{} finds a refactoring of size 630 within 300s, while it takes 2,440 minutes to prove that the optimal refactoring size is 623.
These results indicate that \name{} in practice quickly finds a near-optimal refactoring but takes time to prove optimality.

In conclusion, these results suggest that \name{} scales well on large and diverse programs from real-world problems, thereby answering \textbf{Q2}.

%% file: data/CompressionCompareCPSATv2.tex
\pgfplotstableread{
task knorf maxrefactor
1 0.40 0.72
2 0.36 0.71
3 0.38 0.71
4 0.38 0.70
5 0.36 0.70
6 0.35 0.70
7 0.40 0.70
8 0.35 0.70
9 0.36 0.70
10 0.36 0.70
11 0.39 0.70
12 0.32 0.70
13 0.38 0.70
14 0.40 0.70
15 0.39 0.70
16 0.35 0.69
17 0.35 0.69
18 0.38 0.69
19 0.35 0.69
20 0.32 0.69
21 0.31 0.69
22 0.34 0.69
23 0.34 0.69
24 0.31 0.69
25 0.34 0.69
26 0.37 0.69
27 0.41 0.69
28 0.36 0.69
29 0.31 0.69
30 0.32 0.69
31 0.36 0.69
32 0.37 0.69
33 0.31 0.68
34 0.36 0.68
35 0.37 0.68
36 0.34 0.68
37 0.31 0.68
38 0.31 0.68
39 0.27 0.68
40 0.36 0.68
41 0.33 0.68
42 0.31 0.68
43 0.36 0.68
44 0.31 0.68
45 0.35 0.68
46 0.34 0.68
47 0.31 0.68
48 0.33 0.68
49 0.32 0.68
50 0.32 0.68
51 0.35 0.68
52 0.36 0.68
53 0.34 0.68
54 0.34 0.68
55 0.32 0.68
56 0.36 0.68
57 0.31 0.68
58 0.33 0.68
59 0.32 0.68
60 0.36 0.68
61 0.30 0.68
62 0.30 0.68
63 0.33 0.68
64 0.29 0.68
65 0.35 0.68
66 0.33 0.68
67 0.34 0.68
68 0.30 0.68
69 0.34 0.68
70 0.32 0.68
71 0.35 0.68
72 0.34 0.68
73 0.35 0.68
74 0.29 0.68
75 0.34 0.68
76 0.35 0.68
77 0.38 0.68
78 0.35 0.68
79 0.31 0.68
80 0.32 0.68
81 0.34 0.68
82 0.31 0.68
83 0.29 0.68
84 0.37 0.68
85 0.32 0.68
86 0.30 0.68
87 0.35 0.68
88 0.31 0.68
89 0.34 0.68
90 0.35 0.68
91 0.37 0.68
92 0.34 0.68
93 0.35 0.68
94 0.34 0.68
95 0.31 0.67
96 0.28 0.67
97 0.34 0.67
98 0.31 0.67
99 0.33 0.67
100 0.34 0.67
101 0.36 0.67
102 0.32 0.67
103 0.36 0.67
104 0.33 0.67
105 0.31 0.67
106 0.35 0.67
107 0.32 0.67
108 0.34 0.67
109 0.26 0.67
110 0.32 0.67
111 0.34 0.67
112 0.34 0.67
113 0.30 0.67
114 0.28 0.67
115 0.32 0.67
116 0.30 0.67
117 0.31 0.67
118 0.33 0.67
119 0.33 0.67
120 0.29 0.67
121 0.33 0.67
122 0.38 0.67
123 0.31 0.67
124 0.32 0.67
125 0.32 0.67
126 0.30 0.67
127 0.32 0.67
128 0.34 0.67
129 0.32 0.67
130 0.34 0.67
131 0.34 0.67
132 0.32 0.67
133 0.36 0.67
134 0.34 0.67
135 0.31 0.67
136 0.34 0.67
137 0.33 0.67
138 0.34 0.67
139 0.32 0.67
140 0.32 0.67
141 0.34 0.67
142 0.37 0.67
143 0.31 0.67
144 0.33 0.67
145 0.34 0.67
146 0.33 0.66
147 0.34 0.66
148 0.31 0.66
149 0.28 0.66
150 0.30 0.66
151 0.31 0.66
152 0.29 0.66
153 0.33 0.66
154 0.31 0.66
155 0.33 0.66
156 0.30 0.66
157 0.29 0.66
158 0.32 0.66
159 0.29 0.66
160 0.30 0.66
161 0.32 0.66
162 0.26 0.66
163 0.32 0.66
164 0.33 0.66
165 0.36 0.66
166 0.34 0.66
167 0.32 0.66
168 0.31 0.66
169 0.31 0.66
170 0.30 0.66
171 0.30 0.66
172 0.30 0.66
173 0.33 0.66
174 0.30 0.65
175 0.29 0.65
176 0.29 0.65
177 0.26 0.65
178 0.30 0.65
179 0.28 0.65
180 0.29 0.65
181 0.26 0.65
182 0.31 0.65
183 0.26 0.65
184 0.30 0.65
185 0.27 0.65
186 0.31 0.65
187 0.30 0.65
188 0.28 0.65
189 0.29 0.65
190 0.33 0.65
191 0.29 0.65
192 0.29 0.65
193 0.26 0.64
194 0.31 0.64
195 0.30 0.64
196 0.29 0.64
197 0.30 0.64
198 0.29 0.64
199 0.32 0.63
200 0.27 0.62
}\mydatalegocpsat

\pgfplotstableread{
task knorf maxrefactor
1 0.00 0.56
2 0.00 0.53
3 0.12 0.53
4 0.07 0.53
5 0.00 0.53
6 0.08 0.52
7 0.00 0.52
8 0.00 0.52
9 0.00 0.52
10 0.07 0.51
11 0.01 0.51
12 0.00 0.50
13 0.05 0.50
14 0.00 0.50
15 0.00 0.49
16 0.00 0.49
17 0.00 0.49
18 0.10 0.49
19 0.00 0.49
20 0.12 0.49
21 0.03 0.49
22 0.00 0.48
23 0.10 0.48
24 0.00 0.48
25 0.02 0.48
26 0.00 0.47
27 0.13 0.47
28 0.00 0.47
29 0.00 0.46
30 0.00 0.46
31 0.00 0.46
32 0.00 0.45
33 0.00 0.45
34 0.00 0.44
35 0.10 0.44
36 0.13 0.44
37 0.12 0.44
38 0.02 0.44
39 0.08 0.44
40 0.00 0.43
41 0.00 0.43
42 0.00 0.43
43 0.00 0.43
44 0.13 0.43
45 0.00 0.43
46 0.00 0.43
47 0.00 0.43
48 0.00 0.43
49 0.00 0.43
50 0.00 0.43
51 0.00 0.42
52 0.00 0.42
53 0.10 0.42
54 0.00 0.42
55 0.07 0.41
56 0.15 0.41
57 0.11 0.41
58 0.00 0.41
59 0.00 0.41
60 0.00 0.41
61 0.15 0.41
62 0.11 0.41
63 0.00 0.41
64 0.00 0.40
65 0.00 0.40
66 0.16 0.40
67 0.07 0.40
68 0.00 0.40
69 0.00 0.40
70 0.00 0.40
71 0.00 0.40
72 0.00 0.40
73 0.00 0.40
74 0.00 0.39
75 0.01 0.36
76 0.00 0.35
77 0.00 0.35
78 0.07 0.35
79 0.00 0.34
80 0.00 0.34
81 0.01 0.34
82 0.00 0.34
83 0.00 0.34
84 0.00 0.33
85 0.00 0.33
86 0.11 0.33
87 0.00 0.33
88 0.00 0.33
89 0.01 0.33
90 0.00 0.33
91 0.00 0.32
92 0.00 0.32
93 0.00 0.32
94 0.00 0.32
95 0.05 0.32
96 0.00 0.32
97 0.20 0.31
98 0.11 0.31
99 0.18 0.31
100 0.00 0.31
101 0.00 0.30
102 0.00 0.30
103 0.00 0.29
104 0.00 0.29
105 0.00 0.29
106 0.04 0.29
107 0.00 0.29
108 0.00 0.28
109 0.03 0.28
110 0.10 0.28
111 0.00 0.28
112 0.12 0.28
113 0.00 0.28
114 0.00 0.27
115 0.00 0.27
116 0.00 0.27
117 0.00 0.27
118 0.00 0.27
119 0.00 0.26
120 0.00 0.26
121 0.26 0.26
122 0.11 0.26
123 0.00 0.26
124 0.03 0.25
125 0.00 0.25
126 0.00 0.25
127 0.05 0.25
128 0.00 0.25
129 0.00 0.24
130 0.00 0.24
131 0.00 0.24
132 0.00 0.23
133 0.00 0.22
134 0.00 0.22
135 0.12 0.22
136 0.00 0.21
137 0.07 0.21
138 0.00 0.20
139 0.00 0.20
140 0.08 0.19
141 0.00 0.19
142 0.00 0.19
143 0.11 0.18
144 0.11 0.18
145 0.00 0.18
146 0.00 0.18
147 0.00 0.18
148 0.00 0.17
149 0.11 0.17
150 0.00 0.16
151 0.02 0.16
152 0.12 0.15
153 0.00 0.14
154 0.00 0.14
155 0.00 0.14
156 0.09 0.12
157 0.00 0.11
158 0.00 0.11
159 0.09 0.10
160 0.00 0.10
161 0.00 0.08
162 0.14 0.08
163 0.01 0.08
164 0.00 0.07
165 0.11 0.07
166 0.00 0.06
167 0.00 0.05
168 0.00 0.05
169 0.01 0.03
170 0.00 0.02
171 0.03 0.02
172 0.19 0.02
173 0.00 0.01
174 0.00 0.01
175 0.00 0.01
176 0.00 0.00
177 0.04 0.00
178 0.00 0.00
179 0.00 0.00
180 0.03 0.00
181 0.00 0.00
182 0.12 0.00
183 0.26 0.00
184 0.00 0.00
185 0.05 0.00
186 0.00 0.00
187 0.00 0.00
188 0.00 0.00
189 0.05 0.00
190 0.00 0.00
191 0.25 0.00
192 0.00 0.00
193 0.00 0.00
194 0.00 0.00
195 0.13 0.00
196 0.00 0.00
}\mydatastringcpsat


\begin{figure}[t]
	\begin{subfigure}{0.465\textwidth}
			\begin{tikzpicture}
				\begin{axis}[%
                    width=1.05\textwidth,
                    height=0.425\textwidth,
					ybar interval,
                    ybar=-0.8pt,
					ymin = 0,
					ymax = 0.8,
					xmin = 0,
					xmax = 201,
					ymajorgrids = true,
					major x tick style = transparent,
					xtick=\empty,
					xticklabel=\empty,
                    xlabel near ticks,
                    ylabel near ticks,
                    xlabel={Program},
                    ylabel={Compression Rate},
                    legend image code/.code={
                       \draw [#1] (0cm,-0.1cm) rectangle (0.2cm,0.25cm); },
					legend cell align=left,
					legend style={
						at={(0.5,0.9)},
						anchor=south,
						legend columns = -1
					} 
					]
					\addplot[style={mygreen,fill=mygreen,fill opacity=0.4,bar width=1pt}] table[y=maxrefactor,x=task]{\mydatalegocpsat};
					\addplot[style={myblue,fill=myblue,fill opacity=0.8,bar width=0.5pt}] table[y=knorf,x=task]{\mydatalegocpsat};
					\legend{MaxRefactor,Knorf}
				\end{axis}
			\end{tikzpicture}
		\caption{Lego}
	\end{subfigure}\\
	\begin{subfigure}{0.465\textwidth}
			\begin{tikzpicture}
				\begin{axis}[%
                    width=1.05\textwidth,
                    height=0.425\textwidth,
					ybar interval,
                    ybar=-0.8pt,%
					ymin = 0,
					ymax = 0.8,
					xmin = 0,
					xmax = 197,
					ymajorgrids = true,
					major x tick style = transparent,
					xtick=\empty,
					xticklabel=\empty,
                    xlabel near ticks,
                    ylabel near ticks,
                    xlabel={Program},
                    ylabel={Compression Rate},
                    legend image code/.code={
                        \draw [#1] (0cm,-0.1cm) rectangle (0.2cm,0.25cm); },
					legend cell align=left,
					legend style={
						at={(0.5,0.9)},
						anchor=south,
						legend columns = -1
					} 
					]
					\addplot[style={mygreen,fill=mygreen,fill opacity=0.4,bar width=1pt}] table[y=maxrefactor,x=task]{\mydatastringcpsat};
					\addplot[style={myblue,fill=myblue,fill opacity=0.8,bar width=0.5pt}]table[y=knorf,x=task]{\mydatastringcpsat};
					\legend{MaxRefactor,Knorf}
				\end{axis}
			\end{tikzpicture}
		\caption{Strings}
	\end{subfigure}
	\caption{
 The compression rates of \name{} and \knorf{} (both calling CP-SAT) with a 600s timeout. There are two bars for each horizontal coordinate, representing the results of \name{} (green) and \knorf{} (blue) on the same program.
 }
	\label{fig:comp-cpsatv2}
\end{figure}

%% file: data/CompressionCompareMaxSATv2.tex
\pgfplotstableread{
task knorf maxrefactor
1 0.40 0.72
2 0.36 0.71
3 0.38 0.71
4 0.38 0.70
5 0.36 0.70
6 0.35 0.70
7 0.40 0.70
8 0.35 0.70
9 0.36 0.70
10 0.36 0.70
11 0.39 0.70
12 0.32 0.70
13 0.34 0.70
14 0.31 0.70
15 0.38 0.70
16 0.40 0.70
17 0.37 0.70
18 0.39 0.70
19 0.35 0.69
20 0.35 0.69
21 0.38 0.69
22 0.35 0.69
23 0.32 0.69
24 0.31 0.69
25 0.34 0.69
26 0.35 0.69
27 0.37 0.69
28 0.34 0.69
29 0.34 0.69
30 0.41 0.69
31 0.36 0.69
32 0.31 0.69
33 0.32 0.69
34 0.36 0.69
35 0.37 0.69
36 0.31 0.68
37 0.31 0.68
38 0.36 0.68
39 0.28 0.68
40 0.37 0.68
41 0.34 0.68
42 0.31 0.68
43 0.31 0.68
44 0.27 0.68
45 0.36 0.68
46 0.34 0.68
47 0.33 0.68
48 0.31 0.68
49 0.36 0.68
50 0.32 0.68
51 0.36 0.68
52 0.31 0.68
53 0.31 0.68
54 0.35 0.68
55 0.34 0.68
56 0.35 0.68
57 0.31 0.68
58 0.33 0.68
59 0.34 0.68
60 0.26 0.68
61 0.32 0.68
62 0.34 0.68
63 0.32 0.68
64 0.32 0.68
65 0.35 0.68
66 0.28 0.68
67 0.36 0.68
68 0.34 0.68
69 0.34 0.68
70 0.32 0.68
71 0.36 0.68
72 0.31 0.68
73 0.33 0.68
74 0.32 0.68
75 0.32 0.68
76 0.36 0.68
77 0.30 0.68
78 0.30 0.68
79 0.33 0.68
80 0.29 0.68
81 0.35 0.68
82 0.33 0.68
83 0.34 0.68
84 0.30 0.68
85 0.34 0.68
86 0.32 0.68
87 0.35 0.68
88 0.34 0.68
89 0.33 0.68
90 0.35 0.68
91 0.29 0.68
92 0.34 0.68
93 0.35 0.68
94 0.38 0.68
95 0.35 0.68
96 0.31 0.68
97 0.29 0.68
98 0.32 0.68
99 0.38 0.68
100 0.31 0.68
101 0.34 0.68
102 0.32 0.68
103 0.31 0.68
104 0.32 0.68
105 0.29 0.68
106 0.32 0.68
107 0.37 0.68
108 0.32 0.68
109 0.30 0.68
110 0.35 0.68
111 0.31 0.68
112 0.34 0.68
113 0.36 0.68
114 0.34 0.68
115 0.33 0.68
116 0.34 0.68
117 0.32 0.68
118 0.32 0.68
119 0.37 0.68
120 0.35 0.68
121 0.34 0.68
122 0.33 0.67
123 0.31 0.67
124 0.33 0.67
125 0.34 0.67
126 0.36 0.67
127 0.33 0.67
128 0.32 0.67
129 0.34 0.67
130 0.31 0.67
131 0.30 0.67
132 0.30 0.67
133 0.30 0.67
134 0.31 0.67
135 0.31 0.67
136 0.33 0.67
137 0.33 0.67
138 0.32 0.67
139 0.29 0.67
140 0.30 0.67
141 0.34 0.67
142 0.33 0.67
143 0.31 0.67
144 0.33 0.67
145 0.30 0.67
146 0.29 0.67
147 0.34 0.67
148 0.29 0.67
149 0.30 0.67
150 0.26 0.67
151 0.32 0.67
152 0.34 0.67
153 0.31 0.67
154 0.33 0.67
155 0.36 0.67
156 0.34 0.67
157 0.34 0.67
158 0.32 0.67
159 0.34 0.67
160 0.31 0.67
161 0.33 0.67
162 0.34 0.67
163 0.34 0.66
164 0.28 0.66
165 0.32 0.66
166 0.29 0.66
167 0.32 0.66
168 0.32 0.66
169 0.26 0.66
170 0.28 0.66
171 0.31 0.66
172 0.31 0.66
173 0.30 0.66
174 0.27 0.66
175 0.31 0.66
176 0.30 0.66
177 0.30 0.66
178 0.33 0.66
179 0.30 0.65
180 0.29 0.65
181 0.26 0.65
182 0.31 0.65
183 0.30 0.65
184 0.29 0.65
185 0.26 0.65
186 0.30 0.65
187 0.31 0.65
188 0.26 0.65
189 0.30 0.65
190 0.29 0.65
191 0.30 0.65
192 0.28 0.65
193 0.29 0.65
194 0.30 0.65
195 0.33 0.65
196 0.29 0.65
197 0.29 0.65
198 0.29 0.65
199 0.32 0.64
200 0.27 0.63
}\mydatalegomaxsat

\pgfplotstableread{
task knorf maxrefactor
1 0.00 0.60
2 0.00 0.60
3 0.08 0.59
4 0.00 0.59
5 0.00 0.59
6 0.00 0.59
7 0.05 0.59
8 0.01 0.59
9 0.00 0.59
10 0.00 0.59
11 0.00 0.59
12 0.02 0.59
13 0.00 0.59
14 0.00 0.59
15 0.00 0.59
16 0.00 0.59
17 0.03 0.58
18 0.12 0.58
19 0.00 0.58
20 0.00 0.58
21 0.00 0.58
22 0.20 0.58
23 0.00 0.58
24 0.00 0.58
25 0.00 0.58
26 0.13 0.58
27 0.00 0.58
28 0.00 0.58
29 0.11 0.58
30 0.00 0.58
31 0.00 0.58
32 0.01 0.58
33 0.00 0.58
34 0.00 0.58
35 0.00 0.58
36 0.00 0.58
37 0.07 0.58
38 0.00 0.58
39 0.00 0.58
40 0.00 0.58
41 0.00 0.58
42 0.00 0.58
43 0.00 0.58
44 0.07 0.58
45 0.00 0.58
46 0.00 0.58
47 0.12 0.58
48 0.00 0.58
49 0.00 0.58
50 0.04 0.58
51 0.00 0.58
52 0.00 0.58
53 0.00 0.58
54 0.00 0.58
55 0.12 0.58
56 0.00 0.58
57 0.15 0.58
58 0.11 0.58
59 0.05 0.58
60 0.00 0.58
61 0.12 0.58
62 0.12 0.58
63 0.00 0.58
64 0.01 0.58
65 0.00 0.58
66 0.00 0.58
67 0.00 0.58
68 0.19 0.58
69 0.26 0.58
70 0.13 0.58
71 0.00 0.58
72 0.00 0.57
73 0.00 0.57
74 0.00 0.57
75 0.00 0.57
76 0.00 0.57
77 0.01 0.57
78 0.00 0.57
79 0.00 0.57
80 0.00 0.57
81 0.00 0.57
82 0.12 0.57
83 0.00 0.57
84 0.11 0.57
85 0.00 0.57
86 0.11 0.57
87 0.00 0.57
88 0.00 0.57
89 0.07 0.57
90 0.00 0.57
91 0.00 0.57
92 0.00 0.57
93 0.00 0.57
94 0.00 0.57
95 0.00 0.57
96 0.00 0.57
97 0.00 0.57
98 0.00 0.57
99 0.09 0.57
100 0.00 0.57
101 0.03 0.57
102 0.07 0.57
103 0.00 0.57
104 0.26 0.57
105 0.03 0.57
106 0.00 0.57
107 0.00 0.57
108 0.00 0.57
109 0.00 0.57
110 0.05 0.57
111 0.00 0.57
112 0.02 0.57
113 0.25 0.57
114 0.00 0.57
115 0.00 0.57
116 0.01 0.57
117 0.01 0.57
118 0.00 0.57
119 0.00 0.56
120 0.10 0.56
121 0.00 0.56
122 0.00 0.56
123 0.03 0.56
124 0.10 0.56
125 0.00 0.56
126 0.00 0.56
127 0.00 0.56
128 0.00 0.56
129 0.00 0.56
130 0.00 0.56
131 0.00 0.56
132 0.00 0.56
133 0.00 0.56
134 0.18 0.56
135 0.00 0.56
136 0.00 0.56
137 0.11 0.56
138 0.11 0.56
139 0.00 0.56
140 0.00 0.56
141 0.00 0.55
142 0.00 0.55
143 0.00 0.55
144 0.00 0.55
145 0.00 0.55
146 0.00 0.55
147 0.00 0.55
148 0.08 0.55
149 0.00 0.55
150 0.00 0.55
151 0.00 0.55
152 0.11 0.55
153 0.00 0.55
154 0.14 0.55
155 0.00 0.55
156 0.00 0.55
157 0.00 0.55
158 0.00 0.55
159 0.00 0.55
160 0.11 0.55
161 0.00 0.55
162 0.02 0.55
163 0.00 0.55
164 0.00 0.55
165 0.00 0.55
166 0.00 0.54
167 0.00 0.54
168 0.00 0.54
169 0.00 0.54
170 0.13 0.54
171 0.00 0.54
172 0.00 0.54
173 0.00 0.54
174 0.00 0.54
175 0.00 0.54
176 0.00 0.54
177 0.08 0.53
178 0.03 0.53
179 0.05 0.53
180 0.07 0.53
181 0.00 0.53
182 0.00 0.53
183 0.04 0.53
184 0.00 0.52
185 0.11 0.51
186 0.10 0.51
187 0.12 0.50
188 0.05 0.50
189 0.00 0.50
190 0.10 0.49
191 0.07 0.49
192 0.13 0.49
193 0.09 0.48
194 0.10 0.47
195 0.16 0.47
196 0.15 0.45
}\mydatastringmaxsat


\begin{figure}[ht]
\newcommand\graphsizetwo{.465}
    \begin{subfigure}{\graphsizetwo\textwidth}
            \begin{tikzpicture}
                \begin{axis}[%
                    width=1.05\textwidth,
                    height=0.425\textwidth,
                    ybar interval,
                    ybar=-0.8pt,
                    ymin = 0,
                    ymax = 0.8,
                    xmin = 0,
                    xmax = 201,
                    ymajorgrids = true,
                    major x tick style = transparent,
                    xtick=\empty,
                    xticklabel=\empty,
                    xlabel near ticks,
                    ylabel near ticks,
                    xlabel={Program},
                    ylabel={Compression Rate},
                    legend image code/.code={
                       \draw [#1] (0cm,-0.1cm) rectangle (0.2cm,0.25cm); },
                    legend cell align=left,
                    legend style={
                        at={(0.5,0.9)},
                        anchor=south,
                        legend columns = -1
                    } 
                    ]
                    \addplot[style={mygreen,fill=mygreen,fill opacity=0.4,bar width=1pt}] table[y=maxrefactor,x=task]{\mydatalegomaxsat};
                    \addplot[style={myblue,fill=myblue,fill opacity=0.8,bar width=0.5pt}] table[y=knorf,x=task]{\mydatalegomaxsat};
                    \legend{MaxRefactor,Knorf}
                \end{axis}
            \end{tikzpicture}
        \caption{Lego}
    \end{subfigure}\\
    \begin{subfigure}{\graphsizetwo\textwidth}
            \begin{tikzpicture}
                \begin{axis}[%
                    width=1.05\textwidth,
                    height=0.425\textwidth,
                    ybar interval,
                    ybar=-0.8pt,%
                    ymin = 0,
                    ymax = 0.8,
                    xmin = 0,
                    xmax = 197,
                    ymajorgrids = true,
                    major x tick style = transparent,
                    xtick=\empty,
                    xticklabel=\empty,
                    xlabel near ticks,
                    ylabel near ticks,
                    xlabel={Program},
                    ylabel={Compression Rate},
                    legend image code/.code={
                        \draw [#1] (0cm,-0.1cm) rectangle (0.2cm,0.25cm); },
                    legend cell align=left,
                    legend style={
                        at={(0.5,0.9)},
                        anchor=south,
                        legend columns = -1
                    } 
                    ]
                    \addplot[style={mygreen,fill=mygreen,fill opacity=0.4,bar width=1pt}] table[y=maxrefactor,x=task]{\mydatastringmaxsat};
                    \addplot[style={myblue,fill=myblue,fill opacity=0.8,bar width=0.5pt}]table[y=knorf,x=task]{\mydatastringmaxsat};
                    \legend{MaxRefactor,Knorf}
                \end{axis}
            \end{tikzpicture}
        \caption{Strings}
    \end{subfigure}
    \caption{
    The compression rates of \name{} (calling MaxSAT) and \knorf{} (calling CP-SAT) with a 600s timeout. 
    There are two bars for each horizontal coordinate, representing the results of \name{} (green) and \knorf{} (blue) on the same program.}
    \label{fig:comp-maxsatv2}
\end{figure}

%% file: data/Scalability.tex
\usepgfplotslibrary{statistics}

\pgfplotstableread{
size mean se
200 23.85 6.97
250 67.04 12.18
300 86.07 17.03
350 210.96 30.06
400 421.74 92.53
450 933.00 307.13
500 1318.64 275.96
550 2091.56 378.06
600 3234.33 251.02
650 3272.13 185.70
700 3600.15 0.00
}\mydatascalabilitydrugdrug

\pgfplotstableread{
size mean se
200 5.11 0.59
250 14.36 1.37
300 32.86 3.66
350 66.97 10.29
400 246.80 71.84
450 685.65 138.04
500 1205.72 215.97
550 1777.77 219.35
600 2445.61 216.65
650 2885.49 185.25
700 3292.89 139.28
750 3562.20 37.49
}\mydatascalabilityalzheimer

\pgfplotstableread{
size mean se
200 2.70 0.32
250 7.41 1.17
300 11.67 1.68
350 18.45 3.26
400 72.65 20.19
450 359.35 169.38
500 297.03 172.55
550 492.01 163.83
600 1041.12 306.59
650 1209.14 316.04
700 1794.02 359.37
750 2404.26 314.08
800 2454.85 326.92
850 2883.24 271.69
900 3357.68 164.80
950 3335.70 144.00
1000 3425.31 126.13
}\mydatascalabilitywordnet

\pgfplotstableread{
time stringsmean drugdrugmean alzheimermean
0 1.00 1.00 1.00
60 0.71 0.85 0.84
300 0.57 0.31 0.33
600 0.27 0.25 0.30
900 0.01 0.20 0.29
1200 0.00 0.18 0.26
1800 0.00 0.15 0.24
}\mydataexamplemean

\pgfplotstableread{
time lowerwhisker lowerquartile median upperquartile upperwhisker
0 1.0 1.0 1.0 1.0 1.0
60 0.6875 0.859375 1.0 1.0 1.0
300 0.15625 0.21875 0.28125 0.3515625 0.40625
600 0.125 0.21875 0.25 0.3125 0.34375
900 0.125 0.25 0.3125 0.34375 0.375
1200 0.125 0.21875 0.25 0.2890625 0.34375
1800 0.125 0.1875 0.25 0.28125 0.3125
}\mydataexamplealzheimer

\pgfplotstableread{
time lowerwhisker lowerquartile median upperquartile upperwhisker
0 1.0 1.0 1.0 1.0 1.0
60 0.56363 0.80909 0.91818 0.98181 1.0
300 0.09090 0.163636 0.218181 0.29090 0.34545
600 0.09090 0.159090 0.2 0.25454 0.29090
900 0.05454 0.145454 0.19090 0.236363 0.309090
1200 0.109090 0.145454 0.18181 0.218181 0.236363
1800 0.018181 0.127272 0.16363 0.2 0.309090
}\mydataexampledrugdrug

\pgfplotstableread{
time lowerwhisker lowerquartile median upperquartile upperwhisker
0 1.0 1.0 1.0 1.0 1.0
60 0.25316 0.52201 0.74420 0.92452 1.0
300 0.22747 0.38466 0.64171 0.74326 0.86057
600 0.009582 0.01395 0.27982 0.49753 0.63935
900 0.0 0.00343 0.00486 0.005867 0.00736
1200 0.0 0.0 0.0 0.0 0.0
1800 0.0 0.0 0.0 0.0 0.0
}\mydataexamplestrings

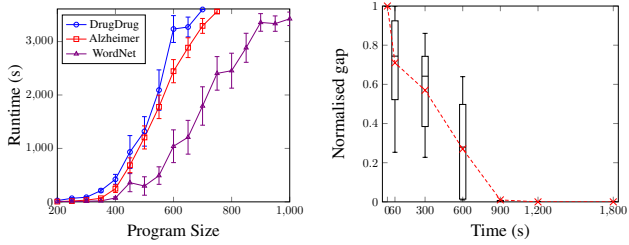
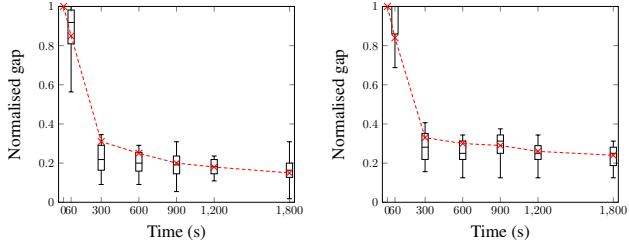
\begin{figure}[!th]
    \centering
    \newcommand\graphsize{0.231}
    \begin{subfigure}[t]{\graphsize\textwidth}
        \resizebox{\linewidth}{!}{
        \begin{tikzpicture}
            \begin{axis}[
                xlabel={Program Size},
                ylabel={Runtime (s)},
                label style={font=\LARGE},
                xmin=200, xmax=1000,
                ymin=0, ymax=3610,
                error bars/y dir=both,
                legend style={at={(0.2,0.97)},anchor=north},
            ]
            \addplot+[blue,mark=o,error bars/.cd, y dir=both, y explicit]
                table[x=size, y=mean, y error=se] {\mydatascalabilitydrugdrug};
                \addlegendentry{DrugDrug}
            \addplot+[red,mark=square,error bars/.cd, y dir=both, y explicit]
                table[x=size, y=mean, y error=se] {\mydatascalabilityalzheimer};
                \addlegendentry{Alzheimer}
            \addplot+[violet,mark=triangle,error bars/.cd, y dir=both, y explicit]
                table[x=size, y=mean, y error=se] {\mydatascalabilitywordnet};
                \addlegendentry{WordNet}
            \end{axis}
        \end{tikzpicture}
        }
        \caption{Average runtime of optimal refactoring versus program size. The bars are standard errors.}
        \label{fig:time-size}
    \end{subfigure}\hfill
    \begin{subfigure}[t]{\graphsize\textwidth}
        \resizebox{\linewidth}{!}{
        \begin{tikzpicture}
            \begin{axis}[
                xlabel={Time (s)},
                ylabel={Normalised gap},
                label style={font=\LARGE},
                xtick={0, 60, 300, 600, 900, 1200, 1800},
                xmin=-40, xmax=1840,
                ymin=0, ymax=1,
                boxplot/draw direction=y,
                boxplot/box extend=-50,
            ]
            \addplot+[boxplot prepared={
              lower whisker=1.0, lower quartile=1.0,
              median=1.0,
              upper quartile=1.0, upper whisker=1.0,
              draw position=0},
              black, solid
            ] coordinates {};
            \addplot+[boxplot prepared={
              lower whisker=0.25316, lower quartile=0.52201,
              median=0.74420,
              upper quartile=0.92452, upper whisker=1.0,
              draw position=60},
              black, solid
            ] coordinates {};
            \addplot+[boxplot prepared={
              lower whisker=0.22747, lower quartile=0.38466,
              median=0.64171,
              upper quartile=0.74326, upper whisker=0.86057,
              draw position=300},
              black, solid
            ] coordinates {};
            \addplot+[boxplot prepared={
              lower whisker=0.00958, lower quartile=0.01395,
              median=0.27982,
              upper quartile=0.49753, upper whisker=0.63935,
              draw position=600},
              black, solid
            ] coordinates {};
            \addplot+[boxplot prepared={
              lower whisker=0.0, lower quartile=0.00343,
              median=0.00486,
              upper quartile=0.005867, upper whisker=0.00736,
              draw position=900},
              black, solid
            ] coordinates {};
            \addplot+[boxplot prepared={
              lower whisker=0.0, lower quartile=0.0,
              median=0.0,
              upper quartile=0.0, upper whisker=0.0,
              draw position=1200},
              black, solid
            ] coordinates {};
            \addplot+[boxplot prepared={
              lower whisker=0.0, lower quartile=0.0,
              median=0.0,
              upper quartile=0.0, upper whisker=0.0,
              draw position=1800},
              black, solid
            ] coordinates {};
            \addplot+[red,mark=x,mark options={solid,mark size=4}]
                table[x=time, y=stringsmean] {\mydataexamplemean};
            \end{axis}
        \end{tikzpicture}
        }
        \caption{Normalised solution quality versus the runtime (a Strings task, size=11,947, opt-time=23min).}
        \label{fig:gap-time-ex1}
    \end{subfigure}\hfill
    
    \begin{subfigure}[t]{\graphsize\textwidth}
        \resizebox{\linewidth}{!}{
        \begin{tikzpicture}
            \begin{axis}[
                xlabel={Time (s)},
                ylabel={Normalised gap},
                label style={font=\LARGE},
                xtick={0, 60, 300, 600, 900, 1200, 1800},
                xmin=-40, xmax=1840,
                ymin=0, ymax=1,
                boxplot/draw direction=y,
                boxplot/box extend=-50,
            ]
            \addplot+[boxplot prepared={
              lower whisker=1.0, lower quartile=1.0,
              median=1.0,
              upper quartile=1.0, upper whisker=1.0,
              draw position=0},
              black, solid
            ] coordinates {};
            \addplot+[boxplot prepared={
              lower whisker=0.56363, lower quartile=0.80909,
              median=0.91818,
              upper quartile=0.98181, upper whisker=1.0,
              draw position=60},
              black, solid
            ] coordinates {};
            \addplot+[boxplot prepared={
              lower whisker=0.09090, lower quartile=0.16363,
              median=0.21818,
              upper quartile=0.29090, upper whisker=0.34545,
              draw position=300},
              black, solid
            ] coordinates {};
            \addplot+[boxplot prepared={
              lower whisker=0.09090, lower quartile=0.15909,
              median=0.2,
              upper quartile=0.25454, upper whisker=0.29090,
              draw position=600},
              black, solid
            ] coordinates {};
            \addplot+[boxplot prepared={
              lower whisker=0.05454, lower quartile=0.145454,
              median=0.19090,
              upper quartile=0.23636, upper whisker=0.30909,
              draw position=900},
              black, solid
            ] coordinates {};
            \addplot+[boxplot prepared={
              lower whisker=0.10909, lower quartile=0.14545,
              median=0.18181,
              upper quartile=0.21818, upper whisker=0.23636,
              draw position=1200},
              black, solid
            ] coordinates {};
            \addplot+[boxplot prepared={
              lower whisker=0.01818, lower quartile=0.12727,
              median=0.16363,
              upper quartile=0.2, upper whisker=0.30909,
              draw position=1800},
              black, solid
            ] coordinates {};
            \addplot+[red,mark=x,mark options={solid,mark size=4}]
                table[x=time, y=drugdrugmean] {\mydataexamplemean};
            \end{axis}
        \end{tikzpicture}
        }
        \caption{Normalised solution quality versus the runtime (a DrugDrug task, size=700, opt-time=2,440min).}
        \label{fig:gap-time-ex2}
    \end{subfigure}\hfill
    \begin{subfigure}[t]{\graphsize\textwidth}
        \resizebox{\linewidth}{!}{
        \begin{tikzpicture}
            \begin{axis}[
                xlabel={Time (s)},
                ylabel={Normalised gap},
                label style={font=\LARGE},
                xtick={0, 60, 300, 600, 900, 1200, 1800},
                xmin=-40, xmax=1840,
                ymin=0, ymax=1,
                boxplot/draw direction=y,
                boxplot/box extend=-50,
            ]
            \addplot+[boxplot prepared={
              lower whisker=1.0, lower quartile=1.0,
              median=1.0,
              upper quartile=1.0, upper whisker=1.0,
              draw position=0},
              black, solid
            ] coordinates {};
            \addplot+[boxplot prepared={
              lower whisker=0.6875, lower quartile=0.85937,
              median=1.0,
              upper quartile=1.0, upper whisker=1.0,
              draw position=60},
              black, solid
            ] coordinates {};
            \addplot+[boxplot prepared={
              lower whisker=0.15625, lower quartile=0.21875,
              median=0.28125,
              upper quartile=0.35156, upper whisker=0.40625,
              draw position=300},
              black, solid
            ] coordinates {};
            \addplot+[boxplot prepared={
              lower whisker=0.125, lower quartile=0.21875,
              median=0.25,
              upper quartile=0.3125, upper whisker=0.34375,
              draw position=600},
              black, solid
            ] coordinates {};
            \addplot+[boxplot prepared={
              lower whisker=0.125, lower quartile=0.25,
              median=0.3125,
              upper quartile=0.34375, upper whisker=0.375,
              draw position=900},
              black, solid
            ] coordinates {};
            \addplot+[boxplot prepared={
              lower whisker=0.125, lower quartile=0.21875,
              median=0.25,
              upper quartile=0.28906, upper whisker=0.34375,
              draw position=1200},
              black, solid
            ] coordinates {};
            \addplot+[boxplot prepared={
              lower whisker=0.125, lower quartile=0.1875,
              median=0.25,
              upper quartile=0.28125, upper whisker=0.3125,
              draw position=1800},
              black, solid
            ] coordinates {};
            \addplot+[red,mark=x,mark options={solid,mark size=4}]
                table[x=time, y=alzheimermean] {\mydataexamplemean};
            \end{axis}
        \end{tikzpicture}
        }
        \caption{Normalised solution quality versus the runtime (an Alzheimer task, size=700, opt-time=1,685min).}
        \label{fig:gap-time-ex3}
    \end{subfigure}
    
    \caption{Scalability results. In (b) to (d), the data points on the red dashed line denote the average normalised gap of 40 independent runs and the boxes show the distribution of these results.}
    \label{fig:scalability}
\end{figure}

%% file: 06-appendix.tex
\section{Proof of Theorem~\ref{thm:complex}}
\label{app:proofcomplex}
\begin{definition}[\textbf{Independent set}]
Let $G=(V,E)$ be a graph. A set of vertices $S\subseteq V$ is an \emph{independent set} if for all $u,v\in S$, $(u,v)\not \in E$. By $I(G)$ we denote the set of independent sets of $G$.  We refer to $S$ as \emph{maximal} if for all $v\in V\setminus S$, $S\cup \{v\}$ is not an independent set. By $I_m(G)$ we denote the set of maximal independent sets of $G$. The \emph{independence number} of a graph $G$, denoted by $\alpha(G)$=$\max \{|S|\mid S\in I_m(G)\}$. The \emph{independence dominating number} of a graph $G$, denoted by $i(G) =\min\{|S|\mid S\in I_m(G)\}$.
\end{definition}

\begin{definition}[\textbf{$\Delta$-regular}]
A graph $G=(V,E)$ is $\Delta$-regular if for every $v\in V$, ¸$\mathit{degree}(v,G)=\Delta$.
\end{definition}

\noindent
Let MIS$^{3}$ be defined as follows:

\vspace{1em}
\fbox{\begin{minipage}{.42\textwidth}
\begin{center}
    \underline{Maximum Independent Set in 3-regular  Hamiltonian}\\  \underline{Graphs (MIS$^{3}$)}
\end{center}
\begin{flushleft}
\underline{\textbf{Given:}} A 3-regular Hamiltonian graph and an integer $k\geq 0$\\
\underline{\textbf{Problem:}} Is there a maximum independent set in $G$ of size $\geq k$.
\end{flushleft}
\end{minipage}}
\vspace{1em}

\citet{FLEISCHNER20102742} proved that MIS$^{3}$ is $\mathcal{NP}$-complete. 

\vspace{1em}
\fbox{\begin{minipage}{.42\textwidth}
\begin{center}
    \underline{Optimal Propositional Knowledge Refactoring} \\ \underline{without Duplication (OPKR$_{WD}$)}
\end{center}
\begin{flushleft}
\underline{\textbf{Given:}} A propositional definite logic program $P$, and a finite set of invented rules $C$\\
\underline{\textbf{Problem:}}  Find a proper propositional refactored program $P'$ such that for any propositional refactored program $Q$, $size(P') \le size(Q)$, and for $C'\subseteq C$,  where $\textsf{unfold}(P' \setminus C', P' \cap C') = P$, the following condition holds for $P'$:
\begin{itemize}
    \item if there exists $c_1,c_2\in C'$ such that $\mathit{body}(c_1)\cap \mathit{body}(c_2)\not = \emptyset$ then for all $r\in P'$, $\{\mathit{pred}(c_1),\mathit{pred}(c_2)\}\not \subseteq \mathit{body}(r)$.
\end{itemize}
\end{flushleft}
\end{minipage}}
\vspace{1em}

\begin{definition}[\textbf{$G$-induced refactoring problem}]
Let $G=(V,E)$ be a 3-regular Hamiltonian graph. Then the \emph{G-Induced Refactoring Problem}, $\mathit{induced}(G)=(P,C,g,h)$, is defined as follows:
\begin{itemize}
    \item $B= \{b_1,\cdots, b_{|E|}\}$ where for all $1\leq i<j\leq |E|$, $b_i\not = b_j$ 
    \item  $C= \{c_1,\cdots, c_{|E|}\}$ where for all $1\leq i<j\leq |E|$, $c_i\not = c_j$
    \item $A=\{a_1,\cdots, a_{|V|}\}$ where for all $1\leq i<j\leq |V|$, $a_i\not = a_j$
    \item $g:V\mapsto A$ and $h:E\mapsto B$ are bijections
    \item $P=\{P_1,\cdots,P_{|E|}\}$ be a set of  propositional definite rules such that for all $1\leq i\leq |E|$, $\mathit{body}(P_i)=B\cup \{c_i\}$.
    \item $C=\{C_1,\cdots,C_{|V|}\}$ be a set of propositional definite rules such that the following holds:
    \begin{itemize}
        \item for all $c,c'\in C$, $\mathit{pred}(c)\not = \mathit{pred}
(c')$, 
        \item for all $c\in C$, $\mathit{pred}(c)\in A$, and
        \item For all $c\in C$, $\mathit{body}(c)\subseteq B$ and  for all $d\in B$, $d\in \mathit{body}(c)$ iff there exists $v\in V$ such that $h^{-1}(d)=(v,g^{-1}(\mathit{pred}(c))) \in E$.
    \end{itemize}
\end{itemize}
\end{definition}

\begin{lemma}
    OPKR$_{WD}$  is $\mathcal{NP}$-hard. 
\end{lemma}
\begin{proof}
We  consider a reduction from MIS$^{3}$. Let $G=(V,E)$ be a 3-regular hamiltonian graph, $\mathit{induce}(G)=(P,C,g,h)$ is a $G$-induced refactoring problem, and $P'$ a solution to OPKR$_{WD}$ for $P$ and $C$ where for   $C'\subseteq C$,  $\textsf{unfold}(P' \setminus C', P' \cap C') = P$ . Observe that for all $c_1,c_2\in C'$, $\mathit{body}(c_1)\cap \mathit{body}(c_2)=\emptyset$ because all $p\in P$, $\bigcup_{c\in C} \mathit{body}(c) \subset \mathit{body}(p)$. Hence,  $S=\{g^{-1}(\mathit{pred}(c))\mid c\in C\}$ forms an independent set of $G$, i.e. $S\in I(G)$. Furthermore, $S\in I_m(G)$, i.e. $S$ is maximal, because if $S\not \in I_m(G)$, then there exists a $W\subset V$ such that $S\cup W\in I_m(G)$  and $C'\cup \{g(w)\mid w\in W\}$ would result in a program $P''$ that is a better refactoring of $P$.  Furthermore, $S$ must be maximum, as otherwise, there would exists $S'\in I_m(G)$ such that $|S'|\geq |S|$, and $C''=\{g(w)\mid w\in S'\}$  would result in a program $P''$ that is a better refactoring of $P$.
\end{proof}
\subsection{Optimal Refactoring With Duplication}
Proposition 2.1 of~\citet{GODDARD2013839} presents a lower bound on the independent dominating number (minimum maximal independent set) of graphs with maximum degree $\Delta$:

\begin{proposition}[\cite{DBLP:books/daglib/0067501,GODDARD2013839}]
\label{prop:IDlower}
    For a graph $G$ with $n$ vertices and maximum degree $\Delta$,
    $$\left\lceil \frac{n}{\Delta+1} \right\rceil\leq i(G)\leq n-\Delta$$
\end{proposition}

\noindent
Observe that this proposition holds for $\Delta$-regular graphs as they are a special case of graphs with maximum degree $\Delta$. Furthermore, Observation 4.1 presented in~\citet{GODDARD2013839}, puts an  upper bound on the size of the maximum independent set of a regular graph.
\begin{observation}[\cite{GODDARD2013839,Rosenfeld1964}]
\label{obs:upperbound}
   If $G$ is a regular graph on $n$ vertices with no isolated vertex, then $i(G) \leq \alpha(G) \leq \frac{n}{2}$.
\end{observation}

\indent
\vspace{1em}
\fbox{\begin{minipage}{.42\textwidth}
\begin{center}
    \underline{Optimal Propositional Knowledge Refactoring} \\ \underline{with Duplication (OPKR)}
\end{center}
\begin{flushleft}
\underline{\textbf{Given:}} A propositional definite logic program $P$, and a finite set of invented rules $C$\\
\underline{\textbf{Problem:}} Find a proper propositional refactored program $P'$ such that for any propositional refactored program $Q$, $size(P') \le size(Q)$.
\end{flushleft}
\end{minipage}}
\vspace{1em}

\begin{lemma}
\label{lem:subsumed}
Let $G=(V,E)$ be a 3-regular Hamiltonian graph and $\mathit{induce}(G)=(P,C,g,h)$ a $G$-induced refactoring problem. Then every solution $P'$ to OPKR with input $P$ and $C$ where for $C'\subseteq C$, $\textsf{unfold}(P' \setminus C', P' \cap C') = P$, $C'$ is a superset of a set $C''\subseteq C$ where $P''$ is solution to OPKR$_{WD}$ with input $P$ and $C$ and $\textsf{unfold}(P'' \setminus C'', P'' \cap C'') = P$. 
\end{lemma}
\begin{proof}
To simplify the proof we will refer to solutions of OPKR$_{WD}$ and OPKR by the subset of $C$ which is needed to unfold the resulting program to the input program. First let us make the following four observations:
\begin{itemize}
    \item[A)]  for any $c_1,c_2\in C$, where $c_1\not = c_2$ and $\mathit{body}(c_1)\cap\mathit{body}(c_2)=\emptyset$, the refactoring of $P$ using $\{c_1,c_2\}$ results in a program $P'$ such that $\mathit{size}(P')=   \mathit{size}(P)- (4|E|-8)$. 
    \item[B)]  for any $c_1,c_2\in C$, where $c_1\not = c_2$ and $|\mathit{body}(c_1)\cap\mathit{body}(c_2)|=1$ the refactoring of $P$ using $\{c_1,c_2\}$ results in a program $P'$ such that $\mathit{size}(P')=   \mathit{size}(P)- (3|E|-8)$.
    \item[C)]  for any $c_1,c_2\in C$, where $c_1\not = c_2$ and $|\mathit{body}(c_1)\cap\mathit{body}(c_2)|=2$ the refactoring of $P$ using $\{c_1,c_2\}$ results in a program $P'$ such that $\mathit{size}(P')=   \mathit{size}(P)- (2|E|-8)$.
    \item[D)]  for any mutually distinct $c_1,c_2,c_3\in C$, where  $|\mathit{body}(c_1)\cap\mathit{body}(c_2)|=1$, $|\mathit{body}(c_1)\cap\mathit{body}(c_3)|=1$, and $\mathit{body}(c_2)\cap\mathit{body}(c_3)=\emptyset$, the refactoring of $P$ using $\{c_1,c_2,c_3\}$ results in a program $P'$ such that $\mathit{size}(P')=   \mathit{size}(P)- (4|E|-12)$.
    
\end{itemize}
These observations imply that a solution $S'$ to OPKR should contain the maximal number of rules in $C$ which have no overlap, i.e. a maximal independent set of $G$. By Proposition~\ref{prop:IDlower}, the subset $S$ of $S'$ corresponding to a maximal independent set of $G$ must have size at least $\left\lceil \frac{|V|}{4} \right\rceil$. 

Now, let $G'=(V',E')$ be the graph derived from $G$ by removing all vertices and edges associated with rules of $S$
and all isolated vertices, i.e. isolated vertices represent rules of $C$ whose body literals are covered by other rules (observation (D)). Then by Observation~\ref{obs:upperbound}, the maximum independent set of $G'$ is bounded from above by $$ \left\lceil\frac{|V|-\left\lceil \frac{|V|}{4} \right\rceil}{2}\right\rceil = \left\lceil\frac{3|V|}{8}\right\rceil.$$ 
While the maximum independent set of $G'$ does not correspond to an optimal refactoring, it upper bounds the size of the optimal refactoring.

Observe that the size of  $|S|=\left\lceil \frac{|V|}{4} +k\right\rceil$ for some $0\leq k\leq |V|-\frac{|V|}{4}$. Thus, 
$$|S'\setminus S| \leq  
\left\lceil\frac{|V|-\left\lceil \frac{|V|+4k}{4} \right\rceil}{2}\right\rceil = \left\lceil\frac{3|V|-4k}{8}\right\rceil.
$$

Furthermore taking the difference of the bound for $k$ and $k+1$ results in the following

$$
\left\lceil\frac{3|V|-4k}{8}\right\rceil-\left\lceil\frac{3|V|-4(k+1)}{8}\right\rceil= \left\lceil\frac{1}{2}\right\rceil = 1.
$$
 
Thus, increasing $k$ by 1 may result in a decrease in the maximum size of $S'\setminus S$ by $1$. Increasing $k$ by 2 will result in a decrease in the maximum size of $S'\setminus S$ by $1$. By the observations highlighted above, this implies a decrease in the size of $P$ by at least $2|E|-4-(|E|-4)= |E|$ when $k$ is increased by 1 and a decrease in the size of $P$ by $4|E|-8-(|E|-4)= 3|E|-4$ when $k$ is increased by 2. Thus, increasing $k$ always results in a better refactoring of $P$. Hence, $S'$ is optimal when it contains a maximum independent set, i.e. $S$ has to be associated with a  maximum independent set of $G$. 
\end{proof}

\begin{lemma}
\label{lem:overlap}
Let $G=(V,E)$ be a 3-regular Hamiltonian graph, $\mathit{induce}(G)=(P,C,g,h)$ a $G$-induced refactoring problem, and $P'$ a solution to OPKR for input $P$ and $C$ where for $S\subseteq C$, $\textsf{unfold}(P' \setminus S, P' \cap S) = P$ . Then for every rule $c\in S$, $|\{b\mid b\in \mathit{body}(c)\cap\mathit{body}(c')\wedge c'\in S \wedge c\not = c'\}|\leq 1$.
\end{lemma}
\begin{proof}
    By the observations presented in the proof of Lemma~\ref{lem:subsumed}, if $\{b\mid b\in \mathit{body}(c)\cap\mathit{body}(c')\wedge c'\in S \wedge c\not = c'\}>1$, then including $c$ actually reduces the quality of the refactoring implying that $S$ is not optimal.
\end{proof}

\begin{lemma}
    OPKR is $\mathcal{NP}$-hard. 
\end{lemma}
\begin{proof}
   We can reduce OPKR$_{WD}$ to OPKR. Let $G=(V,E)$ be a 3-regular Hamiltonian graph, $\mathit{induce}(G)=(P,C,g,h)$ a $G$-induced refactoring problem, and $P'$ an OPKR solution for input P and C where for $S'\subseteq C$, $\textsf{unfold}(P' \setminus S', P' \cap S') = P$.  By Lemma~\ref{lem:subsumed}, we know $S'$ contains a subset $S$ such that for a solution $P''$ of  OPKR$_{WD}$, $\textsf{unfold}(P'' \setminus S, P' \cap S) = P$ and by Lemma~\ref{lem:overlap}, for any rule $c\in S'\setminus S$, $|\{b\mid b\in \mathit{body}(c)\cap\mathit{body}(c')\wedge c'\in S\}|= 1$. Now let 
   $S^* = \{ c\mid  |\mathit{body}(c)\cap\mathit{body}(c')|=1\wedge c,c'\in S \wedge c'\not = c\}$. Observe that $S^*$ has an even number of members. Thus, we can define $S^*=S_1\cup S_2$ such that $S_1\cap S_2=\emptyset$ and $|S_1|=|S_2|$. A solution to OPKR$_{WD}$ is $Q$ such that $\textsf{unfold}(P' \setminus (S\setminus S_1), P' \cap (S\setminus S_1)) = P$  (or  $\textsf{unfold}(P' \setminus (S\setminus S_2), P' \cap (S\setminus S_2)) = P$).
\end{proof}
\noindent
\figurename~\ref{fig:encodingex} shows an example OPKR encoding.

\indent
\vspace{0em}

\fbox{\begin{minipage}{.42\textwidth}
\begin{center}
    \underline{Optimal Knowledge Refactoring (OKR)}
\end{center}
\begin{flushleft}
\underline{\textbf{Given:}} A definite logic program $P$, and a finite set of invented rules $C$\\
\underline{\textbf{Problem:}} Find a refactored program $\mathcal{Q}$ such that for any refactored program $\mathcal{Q'}$, $size(\mathcal{Q}) \le size(\mathcal{Q'})$
\end{flushleft}
\end{minipage}}
\vspace{1em}
\begin{theorem}
    OKR is $\mathcal{NP}$-hard.
\end{theorem}
\begin{proof}
    Observe that OPKR is a subproblem of OKR.
\end{proof}

\begin{figure*}[t]
     \begin{subfigure}[b]{0.5\textwidth}
        \centering
\scalebox{.8}{
   \begin{tikzpicture}
\begin{scope}[every node/.style={circle,thick,draw}]
    \node[shape=circle,draw=black] (A) at (-.5,-5.5) {$a_1$};
    \node[shape=circle,draw=black] (B) at (1.5,-2.5) {$a_2$};
    \node[shape=circle,draw=black] (C) at (1,-4.25) {$a_3$};
    \node[shape=circle,draw=black] (D) at (-2,-4) {$a_4$};
    \node[shape=circle,draw=black] (E) at (3.75,-7.5) {$a_5$};
    \node[shape=circle,draw=black] (F) at (-1.25,-7.5) {$a_6$} ;
    \node[shape=circle,draw=black] (G) at (4.5,-4) {$a_7$} ;
    \node[shape=circle,draw=black] (H) at (2.25,-5.5) {$a_8$} ;
\end{scope}
\begin{scope}[>={Stealth[black]},
              every node/.style={fill=white,circle},
              every edge/.style={draw=red,very thick}]   
    \path[] [-] (A) edge node {$b_1$} (E);
    \path [-](A) edge node {$b_2$} (F);
    \path [-](A) edge node {$b_3$} (H);
    \path [-](B) edge node {$b_4$} (C);
    \path [-](B) edge node {$b_5$} (D);
    \path [-](B) edge node {$b_6$} (G);
    \path [-](C) edge node {$b_7$} (D);
    \path [-](C) edge node {$b_8$} (H);
    \path [-](D) edge node {$b_9$} (F);  
    \path[dotted] [-](E) edge node {$b_{10}$} (F);
    \path [-](E) edge node {$b_{11}$} (G);   
    \path [-](G) edge node {$b_{12}$} (H);   
\end{scope}
\end{tikzpicture}
}
    \caption{A 3-regular Hamiltonian graph.}
    \label{fig:enter-label}
\end{subfigure}
\begin{subfigure}[b]{0.5\textwidth}
\begin{align*}
    \mathcal{P}= &\left\lbrace \begin{array}{c}
    p\leftarrow b_1,\cdots,b_{12},c_1\\
    \vdots\\
    p\leftarrow b_1,\cdots,b_{12},c_{12}\\
\end{array}\right\rbrace\\
    {S}=& \left\lbrace \begin{array}{ll}
    a_1\leftarrow b_1,b_2,b_3 & a_2\leftarrow b_4,b_5,b_6\\
    a_3\leftarrow b_4,b_7,b_8 & a_4\leftarrow b_5,b_7,b_9\\
    a_5\leftarrow b_1,b_{10},b_{11} & a_6\leftarrow b_2,b_9,b_{10}\\
    a_7\leftarrow b_6,b_{11},b_{12} & a_8\leftarrow b_3,b_8,b_{12}\\
    \end{array}\right\rbrace\\
    \text{OPKR}= &\left\lbrace \begin{array}{cc}    a_1\leftarrow b_1,b_2,b_3 & a_3\leftarrow b_4,b_7,b_8\\
    a_4\leftarrow b_5,b_7,b_9 &  a_7\leftarrow b_6,b_{11},b_{12}\\
    \end{array}\right\rbrace
\end{align*}

     \caption{Encoding of the Graph in Figure \ref{fig:enter-label}. }
    \label{fig:encodingGraph}
\end{subfigure}
 \caption{An example OPKR encoding. Dotted edge is not covered by the OPKR.}
     \label{fig:encodingex}
\end{figure*}

\section{Proof of Theorem~\ref{thm:optref}}
\label{app:tightbound}

 Let us consider an invented rule $r \in {S}$ which minimally reduces the size of logic program $\mathcal{P}$ through refactoring. $r$ has the following properties: (i) $\mathit{body}(r)= 2$, (ii) there exists $R \subseteq \mathcal{P}$ such that $|R|= 4$, (iii) refactoring $R$ by $r$ results in a programs $R'$ such that $\mathit{size(R')}\leq \mathit{size(R)}$. We refer to such $r \in {S}$ as $\mathcal{P}$-\textit{minimal}. Note, if $|R|< 4$, then refactoring by $r \in {S}$  ($\mathit{body}(r)= 2$) does not reduce the size of $R$.

 If refactoring of $\mathcal{P}$ using  $S' \subset {S}$, where all $r \in S'$ are $\mathcal{P}$-minimal, results in a program $\mathcal{P}'$ where for every $c' \in \mathcal{P}'$, $\mathit{body}(c')\subseteq \{ \mathit{pred}(r) \mid r \in S'\}$, then each rule $c^*\in \mathcal{P}$ is refactored by at most $\left\lceil \frac{s-1}{2}\right\rceil$ $\mathcal{P}$-minimal invented rules. In the worst case, each invented rule in $S'$ will refactor precisely 4 rules of $\mathcal{P}$, thus we will need $K = \left\lceil \frac{n}{4}\right\rceil\left\lceil \frac{s-1}{2}\right\rceil$.

If we refactor by invented rules which are not $\mathcal{P}$-minimal, but individually reduce the size of $\mathcal{P}$, fewer invented rules are needed to refactor $\mathcal{P}$ into a program $\mathcal{P}'$ with the outlined properties because (i)  $\leq \left\lceil \frac{s-1}{m}\right\rceil$ invented rules with body size $m\geq 2$ are needed to cover the rules in $\mathcal{P}$ and (ii) larger invented rules need to occur $2+\max\{0,4-m\}$ times in $\mathcal{P}$ to reduce the size of $\mathcal{P}$. Observe that $\left\lceil \frac{n}{2+\max\{0,4-m\}}\right\rceil\left\lceil \frac{s-1}{m}\right\rceil\leq \left\lceil \frac{n}{4}\right\rceil\left\lceil \frac{s-1}{2}\right\rceil$ for $m\geq 2$ and thus the right side is an upper bound.

Note that the bound is tight as there exists programs whose optimal refactoring requires precisely $\left\lceil \frac{n}{4}\right\rceil\left\lceil \frac{s-1}{2}\right\rceil$ invented rules.
Consider the following propositional program:
\[
\mathcal{P} = 
\left\{
\begin{array}{l}
    p\leftarrow  q_1,q_2,q_5,q_6,q_7,q_8 \\
    p\leftarrow  q_1,q_2,q_3,q_4,q_5,q_6\\
    p\leftarrow  q_1,q_2,q_3,q_4,q_7,q_8 \\
    p\leftarrow  q_{11},q_{12},q_{15},q_{16},q_{17},q_{18} \\
    p\leftarrow  q_{13},q_{14},q_{15},q_{16},q_{17},q_{18} \\
    p\leftarrow  q_{11},q_{12},q_{13},q_{14},q_{17},q_{18} \\
    p\leftarrow  q_{11},q_{12},q_{13},q_{14},q_{15},q_{16} \\
    p\leftarrow  q_{9},q_{10},q_{11},q_{12},q_{17},q_{18} \\
    p\leftarrow  q_{9},q_{10},q_{13},q_{14},q_{15},q_{16} \\
    p\leftarrow  q_{3},q_{4},q_{7},q_{8},q_{9},q_{10} \\
    p\leftarrow  q_{1},q_{2},q_{5},q_{6},q_{9},q_{10} \\
    p\leftarrow  q_{3},q_{4},q_{5},q_{6},q_{7},q_{8} \\
\end{array}
\right\}
\]
The considered set of invented rules is as follows: 
\[
R = 
\left\{
\begin{array}{ll}
    p_1\leftarrow q_1,q_2 &  p_2\leftarrow  q_3,q_4 \\
    p_3\leftarrow q_5,q_6 &  p_4\leftarrow  q_7,q_8 \\
    p_5\leftarrow q_9,q_{10} &  p_6\leftarrow  q_{11},q_{12} \\
    p_7\leftarrow q_{13},q_{14} &  p_8\leftarrow  q_{15},q_{16} \\
    p_9\leftarrow q_{17},q_{18} & 
\end{array}
\right\}
\]
Observe that 
\begin{equation*}
        \left\lceil \frac{n}{4}\right\rceil \left\lceil \frac{\mathit{size}(c)-1}{2}\right\rceil =  \left\lceil \frac{12}{4}\right\rceil \left\lceil \frac{6}{2}\right\rceil = 9.
\end{equation*}

\noindent
The refactored program is as follows:
\[
\mathcal{Q} = 
\left\{
\begin{array}{ll}
    p\leftarrow p_1,p_3,p_4 & p\leftarrow p_1,p_2,p_3 \\
    p\leftarrow p_1,p_2,p_4 & p\leftarrow p_{6},p_{8},p_{9} \\
    p\leftarrow p_{7},p_{8},p_{9} & p\leftarrow p_{6},p_{7},p_{9} \\
    p\leftarrow p_{6},p_{7},p_{8} & p\leftarrow p_{5},p_{6},p_{9} \\
    p\leftarrow p_{5},p_{7},p_{8} & p\leftarrow p_{2},p_{4},p_{5} \\
    p\leftarrow p_{1},p_{3},p_{5} & p\leftarrow p_{2},p_{3},p_{4} \\
    p_1\leftarrow q_1,q_2 & p_2\leftarrow q_3,q_4 \\
    p_3\leftarrow q_5,q_6 &  p_4\leftarrow q_7,q_8 \\
    p_5\leftarrow q_9,q_{10} &  p_6\leftarrow q_{11},q_{12} \\
    p_7\leftarrow q_{13},q_{14} &  p_8\leftarrow q_{15},q_{16} \\
    p_9\leftarrow q_{17},q_{18} & 
\end{array}
\right\}
\]
The refactoring by $R$ is optimal as any literal in the body of a rule in $\mathcal{P}$ is only contained in the added invented rules of $\mathcal{Q}$. Furthermore, no literal of $\mathcal{P}$ is used more than once by an invented rule.
Even though the refactoring is large with respect to the input program, it only reduced the size of $\mathcal{P}$ by $9$.
Observe that this program is a Balanced Incomplete Block Design (BIBD)~\cite{designtheory}, in particular a (9,12,4,3,1)-design. 

\section{MaxSAT Encoding of \name{}}
\label{app:maxsat}
We describe the MaxSAT encoding of \name{}. Note that it is equivalent to the COP encoding presented in Section~\ref{sec:encoding}. We provide details on implementing the formulation in the Boolean domain.

\paragraph{Decision Variables}
For each possible invented rule $r_k$ for $k \in [1,K]$ and $p \in Pr(\mathcal{P})$, we use a Boolean variable $\mathit{r}_{k,p}^m$ to indicate there are $m$ literals with predicate symbol $p$ in the body of $r_k$.
The domain of $m$ is defined by the maximum number of times $p$ appears in an input rule.
Moreover, the definitions of $\mathit{use}_{c,k}^t$, $\mathit{cover}_{c,a,k}^{t}$, $\mathit{used}_{k}$, and $\mathit{covered}_{c,a}$ are the same as in Section~\ref{sec:encoding}.

\paragraph{Constraints}
First, the variables $r_{k,p}^m$ should be mutually exclusive for different $m$. The constraint is set as follows:
\begin{equation*}
    \forall k,p,m_1,m_2\,(m_1 < m_2),\; {r}_{k,p}^{m_1} \rightarrow \neg {r}_{k,p}^{m_2}.
\end{equation*}
Second, we add the same constraint to ensure the relation between variables $\mathit{use}_{c,k}^t$:
\begin{equation*}
    \forall c,k,t\,(t>1),\; \mathit{use}_{c,k}^t \rightarrow \mathit{use}_{c,k}^{t-1}.
\end{equation*}

\noindent
The other constraints are slightly modified due to changes in the definition of variable $r_{k,p}^m$.
Constraint~(\ref{eq:cons5}) is replaced as follows:
\begin{equation*}
    \forall c, p, k, t, m\,(p {\rm \;not\;in\;} c),\; \neg (use_{c,k}^t \wedge r_{k,p}^m).
\end{equation*}

\noindent
Moreover, Constraint~(\ref{eq:cons6}) is replaced as follows:
\begin{equation*}
    \forall c,a,k,t\,(e(a,p)),\; {cover}_{c,a,k}^{t} \rightarrow \left( use_{c,k}^t \wedge \exists m,\; r_{k,p}^m \right).
\end{equation*}
For Constraint~(\ref{eq:cons7}), we use a modified pseudo-Boolean encoding:
\begin{equation*}
    \forall c,p,k,t,m,\; r_{k,p}^m \rightarrow \left(\sum_{a \in body(c) \wedge e(a,p)} {cover}_{c,a,k}^{t} \right) \le m.
\end{equation*}
Specifically, we utilise a BDD-like encoding \cite{een2006translating} to transform the pseudo-Boolean inequation into MaxSAT constraints.

Finally, the constraints to calculate $\mathit{used}_k$ and $\mathit{covered}_{c,a}$ remain unchanged:
\begin{equation*}
    \forall k,\; \left( \mathit{used}_k \leftrightarrow \exists c,t,\; \mathit{use}_{c,k}^t \right).
\end{equation*}
\begin{equation*}
    \forall c,a,\; \left( \mathit{covered}_{c,a} \leftrightarrow \exists k,t,\; \mathit{cover}_{c,a,k}^{t} \right). 
\end{equation*}

\subsubsection{Objective}
The objective is slightly modified to minimise the following function:
\begin{equation*}
    \label{eq:obj2}
    \sum_{k}{used_{k}} + \sum_{k,p,m}{m \cdot r_{k,p}^m} + \sum_{c,k,t}{use_{c,k}^t} - \sum_{c,a}{covered_{c,a}}.
\end{equation*}
To clarify, this objective can be transformed into unit soft constraints in MaxSAT: (i) Each $\neg used_{k}$ has a weight of 1; (ii) Each $\neg r_{k,p}^m$ has a weight of $m$; (iii) Each $\neg use_{c,k}^t$ has a weight of 1; (iv) Each ${covered}_{c,a}$ has a weight of 1.